\documentclass[10pt, a4paper, UKenglish, twoside]{article}
\usepackage[UKenglish]{babel}
\usepackage[vmargin=40mm,inner=25mm,outer=30mm,headheight=14pt]{geometry}
\usepackage{bbm}
\usepackage[utf8]{inputenc}
\usepackage[T1]{fontenc}
\usepackage{fancyhdr}
\usepackage{amsmath} 
\usepackage[final]{showkeys}
\usepackage{amssymb}

\usepackage{paralist}
\usepackage{hyperref}
\usepackage{amsthm}
\usepackage{color}
\usepackage{mathtools}
\numberwithin{equation}{section}
\usepackage [pdftex] {graphics}
\usepackage{graphicx}
\theoremstyle{plain}
\usepackage[mathscr]{eucal}
\usepackage{bbold}
\newtheorem{lemma}{Lemma}[section]

\newtheorem{theorem}[lemma]{Theorem}
\newtheorem{lemmaanddef}[lemma]{Lemma and Definition}
\newtheorem{proposition}[lemma]{Proposition}
\theoremstyle{definition}

\theoremstyle{remark}
\newtheorem{remark}[lemma]{Remark}
\DeclareMathOperator{\tr}{Tr}
\DeclareMathOperator{\myL}{{\bs{L}}}
\renewcommand{\d}{\mathrm{d}}

\newcommand{\mc}{\mathcal}
\newcommand{\ms}{\mathscr}
\newcommand{\bs}{\boldsymbol}

\newcommand{\bd}{\boldsymbol}
\newcommand{\dx}{\mathrm{d}x}
\newcommand{\dy}{\mathrm{d}y}
\newcommand{\p}{\varphi}
\newcommand{\myM}{{\bs{M}}}
\newcommand{\myK}{{\Xi}}
\newcommand{\R}{{\mathbb{R}}}
\newcommand{\vertiii}[1]{{\left\vert\kern-0.25ex\left\vert\kern-0.25ex\left\vert #1 
    \right\vert\kern-0.25ex\right\vert\kern-0.25ex\right\vert}}

\newcommand{\blue}[1]{{\color{black}{#1}}}

\begin{document}
 \title{Finite Range Decomposition for Gaussian Measures with Improved Regularity}
  \author{Simon Buchholz}
  \date{\today}
  \maketitle
  \begin{abstract}
 We consider a family of gradient Gaussian vector fields on the torus 
 $(\mathbb{Z}/L^N\mathbb{Z})^d$.
Adams, Koteck\'{y},  M\"{u}ller and independently Bauerschmidt established  the existence of a uniform finite range
decomposition of the
corresponding covariance operators, i.e., the covariance can be written as a sum of covariance operators supported on increasing cubes with
diameter $L^k$.
We improve this result and show that the decay behaviour of the kernels in Fourier space  
can be controlled. Then we show the regularity of the integration map that convolves functionals with the partial measures of the finite range
decomposition. 
 In particular the new finite range decomposition avoids the loss of regularity which arises in the renormalisation
 group approach to anisotropic problems in statistical mechanics.
 \end{abstract}
 {\small
{ \bf keywords:} Gradient Gaussian field, finite range decomposition, renormalisation, Fourier multiplier
 }
 
\section{Introduction}\label{secintro}
In this paper we consider finite range decompositions for families of translation invariant Gaussian fields on a torus
${T}_N=(\mathbb{Z}/L^N\mathbb{Z})^d$. 
A Gaussian process $\xi$ indexed by ${T}_N$ has range $M$ if $\mathbb{E}(\xi(x)\xi(y))=0$ for
any $x,y$ such that $|x-y|\geq M$.
A finite range decomposition of $\xi$ is a decomposition $\xi=\sum_k \xi_k$ such that the $\xi_k$ are independent processes
with
range $\sim L^k$.
Equivalently, if $\mc{C}(x,y)$ is the covariance of $\xi$ then a finite range decomposition 	is possible
if there are covariances $\mc{C}_k(x,y)$ such that $\mc{C}=\sum_k \mc{C}_k$, $\mc{C}_k(x,y)=0$ for $|x-y|\gtrsim L^k$, and 
$\mc{C}_k$ is positive semi-definite.

Here we consider vector valued Gaussian fields $\xi_A$ whose covariance is the Greens function of a constant coefficient,
anisotropic, 
elliptic, discrete
difference operator 
$\ms{A}=\nabla^\ast A\nabla$ (plus higher order terms). 
Our main object of interest is the corresponding gradient Gaussian field $\nabla \xi_A$, i.e., we consider the $\sigma$-algebra
generated by the gradients. 
They are referred to as massless field in the language of quantum field theory.
Gradient fields appear naturally in  discrete elasticity where the energy only depends on the distance between the atoms.
The analysis
of gradient Gaussian fields is difficult because they exhibit long 
range correlations only decaying critically 
as $\mathbb{E}(\nabla_i\xi^r(x)\nabla_j\xi^s(y))\propto |x-y|^{-d}$. 
Finite range decompositions of
gradient Gaussian fields are the basis of a multi-scale approach to control
the correlation structure of the fields and avoid logarithmic divergences that appear in naive approaches.

Finite range decompositions of quadratic forms have appeared in different places in mathematics. Hainzl and Seiringer obtained decompositions of
radially symmetric functions into weighted integrals over tent functions \cite{MR1930084}. 
The first decomposition for a setting without radial symmetry was obtained for
the discrete Laplacian by Brydges, Guadagni, and Mitter in \cite{MR2070102}.
Their results are based on averaging the Poisson
kernel.
Brydges and Talaczyck in \cite{MR2240180} generalised this result to
quite general elliptic operators on $\mathbb{R}^m$ that can be written as
$\mc{A}=\mc{B}^\ast\mc{B}$. Adams, Koteck\'{y}, and M\"{u}ller adapted this work in  \cite{adams2013finite} 
to the discrete anisotropic setting. Their decomposition has the property that the kernels $\mc{C}_{A,k}$ are analytic
function of the operator $A$.
Later, Bauerschmidt gave a very general construction based on the finite propagation speed of the wave
equation and functional calculus \cite{MR3129804}.

The goal of this work is to improve the regularity of the previous constructions.
We show lower bounds
for the previous decomposition and modify the construction such that we can control
the decay behaviour of the kernels in Fourier space from above and below.
This implies that the integration map $F\rightarrow \mathbb{E}
(F(\cdot+\xi_{A,k}))$ is differentiable with respect
to the matrix $A$ uniformly in the size $N$ of the torus.
Our results hold for vector valued fields and we allow for higher order terms in the elliptic operator which corresponds 
to general quadratic finite range interaction. This allows us to handle, e.g., realistic models for discrete elasticity where next
to nearest neighbour interactions are included.
The construction is based on the Bauerschmidt decomposition in \cite{MR3129804}
but in a previous version of this project
\cite{frd2} we started from the construction in \cite{adams2013finite}.

The main application of finite range decompositions is the
renormalisation group approach
to problems in statistical mechanics. Renormalisation was introduced by Wilson
in the analysis of phase transitions \cite{MR709078}. Brydges and Yau
\cite{MR1048698} adapted Wilson's ideas to the statistical mechanics setting and initiated a long stream of
developments. Recently Bauerschmidt, Brydges, and Slade introduced a general framework and
investigate the $\phi^4$ model \cite{MR3269689} and the weakly
self avoiding random walk \cite{MR3339164}.
Their approach allows one to analyse
functional integrals $\mathbb{E}(K)$ where $K$ is a non-linear functional depending on a field
on a large lattice and the expectation is with respect to a (gradient) Gaussian field with long ranged correlation.
A key step is that this  integral can be 
rewritten as a series of integrations using a finite range decomposition of the Gaussian field.
Then one can analyse the correlation structure scale by scale.

Adams, Koteck\'y, and M\"uller extend this method to the anisotropic setting 
where $\ms{A}$ is not necessarily a multiple of the Laplacian and they
show strict convexity of the surface tension for non-convex
potentials for small tilt and low temperature \cite{adams2016preprint}. 
However, they face substantial technical difficulties because the integration map 
$F\rightarrow \mathbb{E}
(F(\cdot+\xi_{A,k}))$ is not differentiable with respect to $A$ for their finite range decomposition and
regularity is lost.
 The results of this paper 
allow to avoid this loss of regularity and therefore  simplify their analysis.

This paper is structured as follows: In Section \ref{secsetting} we introduce the setting, state the main result, and
give a brief motivation for the bounds of the new finite range decomposition. Then, in Section \ref{secconstruction}
we prove the main result based on the finite range decomposition from \cite{adams2013finite}.
Finally, in Section \ref{secloss} we show the smoothness of the integration map. 
For the convenience of the reader the appendix states the precise results of \cite{adams2013finite} when applied
to our setting.

Notation: In this paper we always understand inequalities of the form
\begin{align}
 A\geq B
\end{align}
for $A,B\in \mathbb{C}_{\mathrm{her}}^{n\times n}$ in the sense of Hermitian matrices,i.e., 
$A-B$ is semi-positive definite. Moreover we identify in this setting real numbers
with the corresponding multiple of the identity matrix, i.e., $t\in \R$ is
identified with $t\cdot\mathrm{Id}\in  \mathbb{C}_{\mathrm{her}}^{n\times n}$.

\section{Setting and Main Result}\label{secsetting}
Fix an odd integer $L\geq 3$, a dimension $d\geq 2$, and the number of components $m\geq 1$ for the rest of this paper. Let
$T_N=(\mathbb{Z}/(L^N\mathbb{Z}))^d$ be the $d$ dimensional discrete torus of side length $L^N$.
We equip $T_N$ with the quotient distances $d$ (or $|\cdot|$) and $d_\infty$ (or
$|\cdot|_\infty$) induced by the \blue{E}uclidean and  
maximum norm respectively.
We are interested in gradient fields. The condition of being a gradient is, however, complicated in dimension $d\geq 2$. Therefore we work with
usual fields modulo a 
constant which are in one-to-one correspondence to gradient fields.
Define
the space of $m$-component fields as
\begin{align}\label{eq:defofVN}
 \mc{V}_N=\{\varphi: T_N\rightarrow \mathbb{R}^m\}=(\mathbb{R}^m)^{T_N}.
\end{align}
Since we take the quotient modulo constant fields we can restrict our fields to have average zero, hence we set
\begin{align}\label{eqdefofadmissible}
 \mc{X}_N=\left\{\varphi\in \mc{V}_N: \sum_{x\in T_N}\varphi(x)=0\right\}.
\end{align}
Let the dot denote the standard scalar product on $\mathbb{R}^m$ which is later extended to $\mathbb{C}^m$. For $\psi, \varphi\in \mc{X}_N$ the
expression
\begin{align}
 \langle \varphi, \psi\rangle =\sum_{x\in T_N} \varphi(x)\cdot \psi(x)
\end{align} 
defines a scalar product on $\mc{X}_N$ and this turns $\mc{X}_N$ into a Hilbert space $\mc{H}$. 
The discrete forward and backward derivatives are defined by
\begin{align}\begin{split}
(\nabla_j\varphi)_r(x)&=\varphi_r(x+e_j)-\varphi_r(x)\qquad r\in\{1,\ldots, m\}, \quad j\in\{1,\ldots, d\},\\
(\nabla_j^*\varphi)_r(x)&=\varphi_r(x-e_j)-\varphi_r(x) \qquad r\in\{1,\ldots, m\}, \quad j\in\{1,\ldots, d\}.
\end{split}\end{align}
Here $e_j$ are the standard unit vectors in $\mathbb{Z}^d$.
Forward and backward derivatives are adjoints of each other. 

Next we introduce the set of operators for which we discuss finite range
decompositions.
We fix some necessary notation.
Let $\mc{M}\subset \mathbb{N}_0^d\setminus \{0,\ldots,0\}$ be a finite set of
multiindices which is fixed for the rest of this work.
We assume that $\mc{M}$ contains all multiindices of order 1,
i.e., the gradient.
We define $R=\max_{\alpha\in\mc{M}} |\alpha|_\infty$.
With $\mc{G}=(\mathbb{R}^m)^\mc{M}$ we denote the space of discrete derivatives
with $\alpha\in\mc{M}$ and for any field $\p$ we denote $D\p(x)\in \mc{G}$ for the vector
$(\nabla^\alpha \p(x))_{\alpha\in\mc{M}}$. 
We equip $\mc{G}$ with the standard scalar product
\begin{align}
 (D\p(x),D\psi(x))=\sum_{\alpha\in\mc{M}}
(\nabla^\alpha\p(x),\nabla^\alpha\psi(x))_{\mathbb{R}^m}.
\end{align}
For any $z\in\mc{G}$ we write
$z^\nabla$ for the restriction of $z$ to the gradient components $\alpha=e_1,\ldots, e_d$. 
We consider non-negative quadratic forms $Q:\mc{G}\rightarrow \mathbb{R}$ that
satisfy
\begin{align}\label{eq:conditionQ}
Q(z)\geq \omega_0 |z^\nabla|^2 
\end{align}
 for some $\omega_0>0$. This condition for general finite range
interactions already appeared in \cite{MR1461951}.
To keep the notation consistent with \cite{adams2013finite} we 
denote the corresponding symmetric generator by
$A:\mc{G}\rightarrow\mc{G}$. By definition $A$ satisfies $(z,Az)=Q(z)$. 
The matrix elements of $A$ are denoted by
$A_{\alpha\beta}\in \mathbb{R}^{m\times m}$, i.e.,
$Q(D\p(x),D\psi(x))=( AD\p(x),D\psi(x))=\sum_{\alpha,\beta\in\mc{G}}
\nabla^\alpha \p(x) A_{\alpha\beta}\nabla^\beta \psi(x)$.
Usually we consider the set of generators $A$ whose operator norm with respect to the standard scalar
product on $\mc{G}$ satisfies
\begin{align}\label{eq:conditionA}
 \lVert A\rVert\leq \Omega_0
\end{align}
for some fixed $\Omega_0>0$.
We denote the set of symmetric operators $A$ such that \eqref{eq:conditionQ}
and \eqref{eq:conditionA} hold by $\mc{L}(\mc{G},\omega_0,\Omega_0)$. We think of $\mc{M}$,
 $\omega_0$, and 
$\Omega_0$ as fixed and in the following and all constants
depend on $d$, $m$, $R$, $\omega_0$, and $\Omega_0$ in the following.
From the operator $A$ we obtain a corresponding elliptic finite difference
operator
\begin{align}
\ms{A}=\sum_{\alpha,\beta\in \mc{G}} (\nabla^\alpha)^\ast
A_{\alpha\beta}\nabla^\beta.
\end{align}
The operator $\ms{A}$ defines a Gaussian measure $\mu_A$ on $\mc{X}_N$ that is given by
\begin{align}\label{eq:defGaussmeasure}
 \mu_A(\d \p)=\frac{e^{-\frac12\langle \p,\ms{A}\p \rangle}}{\sqrt{\det\left(2\pi \ms{A}^{-1}\right)}}\lambda(\d\p)
\end{align}
where $\lambda$ denotes the $m(L^{Nd}-1)$ dimensional Hausdorff-measure on the affine space $\mc{X}_N$.

In the following we discuss finite range
decompositions for the Greens functions of operators $\ms{A}:\mc{X}_N\rightarrow\mc{X}_N$ 
with $A\in \mc{L}(\mc{G},\omega_0,\Omega_0)$.
In \cite{MR3129804}  and \cite{adams2013finite} only the case where
$\mc{M}=\{e_1,\ldots, e_d\}$ was discussed. 
By assumption $\langle \p,\ms{A}\p\rangle \geq \omega_0\langle
\nabla\p,\nabla\p\rangle$ which implies positivity of the operator.
Hence the operator $\mathscr{A}$ is invertible and we call its
inverse operator $\mathscr{C}=\mathscr{A}^{-1}$. The operator $\ms{C}$ is the covariance operator of the Gaussian measure $\mu_A$. Since $\mathscr{A}$ 
is translation invariant (i.e.,  $[\tau_x, \mathscr{A}]=0$ where
$\tau_x:\mc{X}_N\rightarrow \mc{X}_N$ denotes the translation operator $(\tau_x\phi)(y)=\phi(y-x)$ and $[A,B]=AB-BA$ \blue{is} the commutator) the
same is true for $\mathscr{C}$. 
Translation invariance implies that the operator $\mathscr{C}$ has a unique kernel  $\mc{C}:T_N\rightarrow \mathbb{R}^{m\times m}$ (cf. Lemma 3.5
in \cite{adams2013finite}) such that
\begin{align}\label{defofkernel}
 (\mathscr{C}\p)(x)=\sum_{y\in T_N} \mc{C}(x-y)\p(y)
\end{align}
and $\mc{C}\in \mc{M}_N$ where $\mc{M}_N$ is the space of $m\times m$
matrix valued functions on $T_N$ with average zero, i.e., $\mc{C}_{ij}\in \mc{X}_N$ for all $1\leq i,j\leq m$ (this condition can always be
satisfied because constant kernels generate the zero-operator on $\mc{X}_N$).

\begin{remark}\label{remark:VnvsXn}
Note that the kernel  $\mc{C}$ defines an extension $\overline{\ms{C}}$ of the operator $\ms{C}$ from the space $\mc{X}_N$
to the space $\mc{V}_N$ that annihilates constant fields. Since the space of constant fields is the orthogonal complement of 
$\mc{X}_N$ in $\mc{V}_N$ the operator $\overline{\ms{C}}$ is a positive semi-definite operator on an euclidean space. Hence it is the covariance of a Gaussian measure on $\mc{V}_N$ and this measure 
is concentrated on $\mc{X}_N$ and its restriction to $\mc{X}_N$ agrees with $\mu_A$ given by equation \eqref{eq:defGaussmeasure}.
This implies by general Gaussian calculus that
\begin{align}
 \mathbb{E}_{\mu_A} \left(\p(x)\p(y)\right) = \mc{C}(x-y).
\end{align}
\end{remark}

Let us define the term {\it finite range}.
\begin{lemmaanddef}\label{deffiniterange}
Let $\ms{C}:\mc{X}_N\rightarrow \mc{X}_N$ be a translation invariant operator with kernel $\mc{C}\in\mc{M}_N$. We say that $\ms{C}$ has range at
most $l$ for $2l+3<L^N$ if the following three equivalent statements hold.
\begin{enumerate}
\item $\langle \p,\ms{C}\psi\rangle=0$ for all $\p,\psi\in \mc{X}_N$ with $\textrm{dist}_\infty(\textrm{supp}\,\p,\textrm{supp}\, \psi)>l$.
\item There is $M\in \textrm{Mat}_{m,m}(\mathbb{R})$ such that $\mc{C}(x)=M$ for $\d_\infty(x,0)>l$.
\item $\mathrm{supp}\, \ms{C}\p\subset \mathrm{supp}\,\p+\{-l,\ldots,l\}^d$ for all $\p\in\mc{X}_N$. 
\end{enumerate}
For $2l+3\geq L^N$ property $2$ shall be the defining property.
\end{lemmaanddef}
\begin{proof}
This is Lemma 3.6 in \cite{adams2013finite}.  The implication $(ii)\Rightarrow (iii)\Rightarrow (i)$ is always true. 
\end{proof}
Observe that the operator $\ms{A}$ introduced before has range at most $R$.
We seek a decomposition of $\mathscr{C}$ into translation invariant and positive operators $\mathscr{C}_k$ with
\begin{align}\label{eqcovariancedecomp}
 \mathscr{C}=\sum_{k=1}^{N+1} \mathscr{C}_k
\end{align}
such that the range of $\mathscr{C}_k$ is smaller than ${L^k}/{2}$ and the $\mathscr{C}_k$ satisfy certain bounds. This implies that $\ms{C}_k$ is the covariance
operator of a Gaussian measure on $\mc{X}_N$ and the 
Gaussian variables $\nabla_i\p(x)$ and $\nabla_j\p(y)$ (where $\p$ is distributed according to this measure) are uncorrelated and therefore
independent for $|x-y|_\infty\geq L^k/{2}$.
A decomposition satisfying the finite range property and translation invariance will be called finite range decomposition. Note, however, that
they are only useful in the presence of strong bounds because, e.g., \blue{the trivial decomposition}
$\ms{C}_{N+1}=\ms{C}$ and $\ms{C}_k=0$ for $k\leq N$ has the finite range property.

Translation invariant operators are diagonal in Fourier space which will be introduced briefly because the
strongest bounds of the kernel $\mc{C}_k$ are the bounds for its Fourier transform.
Define the dual torus
\begin{align}
 \widehat{T}_N=
 \left\{\frac{(-L^N+1)\pi}{L^N},\frac{(-L^N+3)\pi}{L^N},\ldots,
\frac{(L^N-1)\pi}{L^N}\right\}^d.
\end{align}
For $p\in\widehat{T}_N$ the exponentials $f_p:T_N\rightarrow \mathbb{C}$ with $f_p(x)=e^{ip\cdot x}$ are well defined since $e^{ip\cdot x}$ is a
$(L^N\mathbb{Z})^d$ periodic function on $\mathbb{Z}^d$. Here and in the following we use 
the immediate generalisations to complex valued fields.
The Fourier transform $\widehat{\psi}:\widehat{T}_N\rightarrow \mathbb{C}$ of a scalar field $\psi:T_N\rightarrow \mathbb{C}$ is defined by
\begin{align}
 \widehat{\psi}(p)=\sum_{x\in T_N} e^{-ip\cdot x}\psi(x)=\sum_{x\in T_N}f_p(-x)\psi(x)
\end{align}
and the inverse transform is given by
\begin{align}\label{inverseFourier}
 \psi(x)=\frac{1}{L^{Nd}}\sum_{p\in \widehat{T}_N}e^{ip\cdot x}\widehat{\psi}(p).
\end{align}
The Fourier transform maps the space $\mc{X}_N$ bijectively on the subspace $\{\widehat{\psi}:\widehat{T}_N\rightarrow
\mathbb{C}|\,\widehat{\psi}(0)=0\}$.
Clearly $\widehat{f}_p(q)=L^{Nd}\delta_{pq}$ for $p,q\in \widehat{T}_N$. For matrix- or vector-valued functions
we define the Fourier transform component-wise.
The Fourier transform satisfies
\begin{align}
 \langle \psi, \varphi\rangle=\frac{1}{L^{Nd}}\sum_{p\in \widehat{T}_N} {\widehat{\psi}}(p)\cdot\widehat{\varphi}(p)
\end{align}
where we extended the scalar product anti-linearly in the first component, this means $v\cdot w=\sum_{i=1}^m \overline{v}_iw_i$. Hence the
functions $L^{-\frac{Nd}{2}}f_pe_i$ for $p\neq 0$ and $e_i\in\mathbb{R}^m$ a standard unit vector form an orthonormal basis of
$\mc{X}_N$ (complexified).
Let $\mc{A},\mc{B}\in \mc{M}_N$ be two matrix-valued functions. Define the convolution by
\begin{align}\label{defconvolution}
\mc{A}\ast\mc{B}(x)=\sum_{y\in T_N} \mc{A}(x-y)\mc{B}(y).
\end{align}
Note that by \eqref{defofkernel} the kernel of the composition  $\ms{A}\ms{B}$ of the operators $\ms{A}$ and $\ms{B}$ with kernels $\mc{A}$ and
$\mc{B}$ is given by $\mc{A}\ast\mc{B}$.

Consider a translation invariant operator $\ms{K}$ with kernel $\mc{K}\in\mc{M}_N$. 
As in the continuous setting the Fourier transform of a convolution is the product
of the Fourier transforms, i.e.,
\begin{align}\label{convolutionofkernel}
 \widehat{\ms{K}\psi}(p)=\widehat{\mc{K}\ast \psi}(p)=\widehat{\mc{K}}(p)\widehat{\psi}(p).
\end{align}
Hence translation invariant operators are indeed (block) diagonal in Fourier space with
eigenvalues given by the Fourier transform of the kernel.

Next we calculate the Fourier modes of the kernel of the operator $\ms{A}$. 
A simple calculation shows that $\widehat{\nabla_i\p}(p)=q_i(p)\widehat{\p}(p)$
and
$\widehat{\nabla_i^\ast\p}(p)=\overline{q}_i(p)\widehat{\p}(p)$ where
$q_i(p)=e^{ip_i}-1$ and $\overline{q}_i(p)=e^{-ip_i}-1$ denotes the complex conjugate.
This implies
\begin{align}
 \widehat{\mc{A}}(p)=\sum_{\alpha,\beta\in\mc{M}} \overline{q}(p)^\alpha
A_{\alpha\beta} q(p)^\beta
\end{align}
where $q(p)^\alpha=\prod_{i=1}^d q_i(p)^{\alpha_i}$.
Note 
again the formal similarity to the continuum setting.
The estimate $\frac{4}{\pi^2}t^2\leq |e^{it}-1|^2 \leq t^2$ for $t\in [-\pi,\pi]$ immediately implies that $\frac{4}{\pi^2}|p|^2\leq |q(p)|^2\leq
|p|^2$ for 
any $p\in \widehat{T}_N$. 
Hence we find using $|p|<\sqrt{d}\pi$ and $|\alpha|_1\leq dR$ for
$\alpha\in\mc{M}$
\begin{align}
  \lVert \widehat{\mc{A}}(p)\rVert=\left\Vert \sum_{\alpha,\beta\in\mc{M}}
\overline{q}(p)^\alpha q(p)^\beta  A_{\alpha\beta}\right\Vert
\leq \lVert A\rVert \left|
\left(q(p)^\alpha\right)_{\alpha\in\mc{M}}\right|^2_{\mc{G}}
\leq \Omega_0 |p|^2 \cdot|\mc{M}| (d\pi^2)^{dR}
\end{align}
On the other hand we also find a lower bound for the Fourier modes and
$a\in\mathbb{R}^m$ using the assumption \eqref{eq:conditionQ} for $z=(aq(p)^\alpha)_{\alpha\in\mc{M}}$ 
\begin{align}
 a\cdot\widehat{\mc{A}}(p)a\geq \omega_0|q(p)|^2|a|^2\geq
\frac{4\omega_0}{\pi^2}|p|^2|a|^2.
\end{align}
Together this yields the important bound
\begin{align}\label{Ahatestimate}
 \lVert \widehat{\mc{A}}(p)\rVert\leq  \Omega |p|^2 \quad \text{ and }
\quad \widehat{\mc{A}}(p)\geq  \omega |p|^2.
\end{align}
Finally we note that by the Fourier inversion formula  \eqref{inverseFourier} for any multiindex $\alpha$ 
\begin{align}\label{derivativeoffourier}
 \nabla^\alpha\mc{C}_k(x)=\frac{1}{L^{Nd}}\sum_{p\in \widehat{T}_N} \widehat{\mc{C}_k}(p)q(p)^\alpha f_p(x).
\end{align}
The $L^\infty-L^1$ bounds for the Fourier transform also hold in the discrete case
\begin{align}\label{l1linftyFourier}
\lVert\nabla^\alpha\mc{C}_k(x)\rVert \leq\frac{1}{L^{Nd}}\sum_{p\in \widehat{T}_N} \lVert \widehat{\mc{C}}_k(p)\rVert |p|^{|\alpha|}.
\end{align}
Finally we introduce a dyadic partition of the dual torus for $j=1,\ldots,N$
\begin{align}\label{eq:defofannulus}
 \bd{A}_j&=\{p\in\widehat{T}_N:\, L^{-j-1}<|p|\leq L^{-j}\}\\
 \bd{A}_0&=\{p\in\widehat{T}_N:\, L^{-1}<|p|\}.
\end{align}
When the size of the torus $\widehat{T}_N$ is not clear from the context we write
$\bd{A}_j^N=\bd{A}_j$ for clarity.
Observe that 
\begin{align}\label{eq:sizeofAj}
 |\bs{A}_j|\leq \kappa(d) L^{(N-j)d}
\end{align}
for some constant $\kappa(d)>0$.

Let us now state the main result of \cite{MR3129804} adapted to our setting.
\begin{theorem}\label{AKMFRD}
For $d\geq 2$ and $A\in \bigcup_{\omega_0>0}\mc{L}(\mc{G},\omega_0, \Omega_0)$
there is a decomposition
$\ms{C}_A=\sum_{k=1}^{N+1} \ms{C}_{A,k}$ where $\ms{C}_{A,k}$ are positive,
translation invariant operators and the map $A\rightarrow \mc{C}_{A,k}(x)$
only depends on $R$ and $\Omega_0$ but not on $\omega_0$. The $\mc{C}_{A,k}$ are polynomials in ${A}$ for $1\leq k\leq N$ and  real analytic for $k=N+1$. This decomposition has the
finite range property
\begin{enumerate}
\item $\mc{C}_{A,k}(x)=-M_k$ for $1\leq k\leq N$ and $|x|_\infty\geq {L^k}/{2}$ where $M_k\in
\mathbb{R}^{m\times m}$ are positive semi-definite matrices independent of
$A$ (for $\Omega_0$ and $\mc{G}$ fixed). In particular $\ms{C}_{A,k}$ has range smaller
${L^k}/{2}$.
\end{enumerate}
Moreover we have the following bounds for $\omega_0>0$ and $A\in\mc{L}(\mc{G},\omega_0, \Omega_0)$  
\begin{enumerate}
\setcounter{enumi}{1} 
\item In Fourier space  the following bounds hold for any positive integer
$\ell$ and $n$ and symmetric $\dot{A}\in \mc{L}(\mc{G})$
\begin{align}\label{fourierbounds}
\sup_{\lVert \dot{A}\rVert\leq 1} \lVert D^\ell_A\widehat{\mc{C}}_{A,k}(p)(\dot{A},\ldots,\dot{A})\rVert\leq
  \begin{cases}
   C_{\ell,\bar{n}}|p|^{-2}(|p|L^{(k-1)})^{-\bar{n}}
\; &
\text{for $|p|>L^{-k}$ ($p\in \bs{A}_j$, $j\leq k-1$),} \\
   C_\ell L^{2k}\lVert\dot{A}\rVert^\ell &\text{for $p\leq L^{-k}$ ($p\in \bs{A}_j$, $j\geq k$).}
 \end{cases}
\end{align}
Here $C_{\ell}$ and $C_{\ell,\bar{n}}$ are constants that do not depend on $L$, $N$, or $k$.
The corresponding lower bound reads
\begin{align}\label{keylowerbound}
\widehat{\mc{C}}_{A,k}(p)\geq  
 \begin{cases}
   c\min\left(|p|^{-2},L^2\right)\quad &\text{for}\quad
k=1,\\
  cL^{2k}\quad &\text{for $|p|<L^{-k}$ ($p\in\bs{A}_j$, $j\geq k$)}.
 \end{cases}
\end{align}
for some constant $c>0$ depending on the same quantities as $C_\ell$.
 \item In particular, we have
\begin{align}\label{discretebounds}
\sup_{x\in T_N}\sup_{\lVert\dot{A}\rVert\leq 1}\left\lVert\nabla^\alpha
D_A^\ell\mc{C}_{A,k}(x)(\dot{A},\ldots,\dot{A})\right\rVert\leq
\begin{cases}
C({\alpha},\ell)
L^{-(k-1)(d-2+|\alpha|)}\;&\text{for}\;d+|\alpha|>2\\
C({\alpha},\ell)\ln(L)
L^{-(k-1)(d-2+|\alpha|)}\;&\text{for}\;d+|\alpha|=2.
\end{cases}
\end{align}
Here $C(\alpha,\ell)$ denotes a constant that does not depend on $L$, $N$, and $k$.
\end{enumerate}
\end{theorem}
This is basically Theorem 1.2 in \cite{MR3129804} except for the simple lower
bound. 
For the convenience of the reader we include all relevant calculations for the
concrete setting in the appendix. The estimates are similar to the ones that
appeared in \cite{MR3332940}.

Our main result is an extension of this result which additionally gives controlled decay of the kernels in Fourier space. In particular the
operators $\ms{C}_{k,A}$ and $\ms{C}_{k,A'}$ are comparable for the new construction.
The main application of the following theorem is the regularity of the
renormalisation map which is stated in Proposition \ref{proprenormmap} in
Section \ref{secloss}.
\begin{theorem}\label{finalFRD}
Let $L>3$
odd, $N\geq 1$ as before and let $\tilde{n}>n$ be two integers. Fix 
$\Omega_0>\omega_0>0$ and consider the family of 
symmetric, positive operators $A\in \mc{L}(\mc{G},\omega_0,\Omega_0)$.
Then there
exists
 a family of finite range decomposition $\ms{C}_{A,k}$ of $\ms{C}_{A}$ such that 
 the
 map $A\rightarrow \mc{C}_{A,k}(x)$
only depends on $R$, $\omega_0$, and $\Omega_0$. The map is a polynomial in ${A}$ for $1\leq k\leq N$ and real analytic for $k=N+1$. The operators $\ms{C}_{A,k}$ satisfy
 \begin{align}\begin{split}
 & \ms{C}_{A}=\sum_{k=1}^{N+1} \ms{C}_{A,k},\\
 & \ms{C}_{A, k}(x)=M_k\, \, \text{for  $1\leq k\leq N$, and }\, |x|_\infty\geq \frac{L^k}{2},
 \end{split}\end{align}
 where $M_k\leq 0$ and $M_k$ is independent of $A$.
 The $\alpha$-th discrete derivative for all $\alpha$ with
$|\alpha|_1\leq n$ is bounded by
\begin{align}\label{eq:discretebounds}
\sup_{x\in T_N}\sup_{\lVert\dot{A}\rVert\leq 1}\left\lVert\nabla^\alpha
D_A^\ell\mc{C}_{A,k}(x)(\dot{A},\ldots,\dot{A})\right\rVert\leq
\begin{cases}
3C({\alpha},\ell)
L^{-(k-1)(d-2+|\alpha|)}\;&\text{for}\;d+|\alpha|>2\\
3C({\alpha},\ell)\ln(L)
L^{-(k-1)(d-2+|\alpha|)}\;&\text{for}\;d+|\alpha|=2.
\end{cases}
\end{align}
  where the constants are the same as in Theorem \ref{AKMFRD}.
 We have the following lower bounds in Fourier space for some
$c=c(\tilde{n})>0$  
  \begin{align}\begin{split}\label{finalfrdlower}
 \widehat{\mc{C}}_{A,k}(p)&\geq 
\begin{cases}	cL^{-2(d+\tilde{n})-1}L^{2j}L^{(k-j)(-d+1-n)}&\text{ for }
p\in \bd{A}_j\text{ and } j< k\\
cL^{-2(d+\tilde{n})-1}L^{2k}&\text{ for } p\in \bd{A}_j\text{ and }
j\geq k.
\end{cases}
 \end{split}\end{align}
Similar upper bounds hold with some constant $C=C(\tilde{n})$
\begin{align}
\begin{split} \label{finalfrdupper}
\lVert\widehat{\mc{C}}_{A,k}(p)\rVert&\leq 
 \begin{cases}
CL^{2(d+\tilde{n})+1}L^{2j}L^{(k-j)(-d+1-n)}&\text{ for } p\in \bd{A}_j\text{ and
} j< k\\
CL^{2k}&\text{ for } p\in \bd{A}_j\text{ and } j\geq k.
\end{cases}
 \end{split}
\end{align}
 For the derivatives of the kernels
and  
for $\lVert \dot{A} \rVert\leq 1$, $\ell\geq 1$, $p\in A_{j}$ we  have the following stronger bounds in Fourier space
\begin{align}
 \left\lVert \frac{\d^\ell}{\d s^\ell}\widehat{\mc{C}}_{A+s\dot{A},k}(p)\right\rVert\leq
\begin{cases}
CL^{2(d+\tilde{n})+1}L^{2j}L^{(k-j)(-d+1-\tilde{n})}&\text{ for } p\in
\bd{A}_j\text{ and } j< k\\
CL^{2k}&\text{ for } p\in \bd{A}_j\text{ and } j\geq k,
\end{cases}
 \end{align}
 i.e.,  the decay of the derivative of the Fourier modes for large $p$ is governed by $\tilde{n}$ and not by $n$ as in \eqref{finalfrdupper}. 
The lower and upper bound can be combined to give for $\ell\geq 1$ and $p\in A_j$
\begin{align}\label{keyquotientbound}
\left\lVert \frac{\d^\ell}{\d s^\ell}\widehat{\mc{C}}_{A+s\dot{A},k}(p)\right\rVert\cdot
\left\lVert \widehat{\mc{C}}_{A,k}(p)^{-1}\right\rVert\leq 
\begin{cases}
\myK L^{4(d+\tilde{n})+2}L^{(k-j)(n-\tilde{n})}  \;&\text{ for } p\in \bd{A}_j\text{
and } j< k\\
\myK L^{2(d+\tilde{n})+1}	  \;&\text{ for } p\in \bd{A}_j\text{ and
} j\geq k.
\end{cases}
\end{align}
The constants $\myK=\myK(\tilde{n},\ell)$ do not depend on $N$, $k$, or $L$.
\end{theorem}

The proof of this theorem can be found in Section \ref{secconstruction}.
\begin{remark}\label{remarkfortheorem}
\begin{enumerate}
\item For the calculations it is advantageous to express the bounds mostly in
terms of the single quantity $L$. To get a better feeling for the bounds it is useful to write them in terms of $|p|$ and $L$.
The definition of $\bs{A}_j$ implies that $|p|\approx L^{-j}$ for $p\in \bs{A}_j$. Hence
up to constants that also depend on $L$ we find $\widehat{\mc{C}}_k(p)\approx
L^{2k}$ for $|p|\lesssim L^{-k}$ and 
$\widehat{\mc{C}}_k(p)\approx |p|^{-2}(L^k|p|)^{-(d-1+n)}$ otherwise.
For $\ell\geq 1$, however,
we find $\lVert \frac{\d^\ell}{\d A^\ell}\widehat{\mc{C}}_{A,k}(p)\rVert\lesssim
|p|^{-2}(L^k|p|)^{-d+1-\tilde{n}}$ for $|p|\gtrsim L^{-k}$ and $\lVert
\frac{\d^\ell}{\d A^\ell}\widehat{\mc{C}}_{A,k}(p)\rVert\lesssim L^{2k}$ otherwise.
In particular the quotient of the derivative and the kernel itself is bounded by a constant
for $|p|\lesssim L^{-k}$ and behaves as $(|p|L^k)^{n-\tilde{n}}$ for $p\gtrsim L^{-k}$, i.e.,  decays
 as fast as we like if we choose $\tilde{n}\gg n$.
\item
 It is possible to obtain similar results using the decomposition based
on averaging the Poisson kernel over cubes that appeared in
\cite{adams2013finite}. Some steps can be found in the earlier version \cite{frd2} of this
work.
To prove Theorem \ref{finalFRD}, however, technical modifications of the construction 
in \cite{adams2013finite} must be implemented in order to handle the generalisation to higher order operators and to get rid of some $L$ dependent constants in the
bound \eqref{discretebounds}. 
\end{enumerate}
\end{remark}

Since the theorem is
rather technical we briefly motivate the need for lower bounds
and the specific structure of the bounds.
As pointed out before we are interested in  bounds for $\partial_A \mathbb{E}_{\ms{C}_{k,A}} (F(\cdot+\xi))$.
 By Theorem \ref{AKMFRD} the derivatives of the covariance $\ms{C}_{k,A}$ with
respect to $A$ are controlled. 
 By the chain rule we need to bound the derivatives of expectations of Gaussian
random variables with respect to their covariance.
Let us briefly discuss this problem in a general setting. Consider a smooth map
$M:\mathbb{R}\rightarrow \mathbb{R}^{s\times
s}_{\textrm{sym},+}$ mapping to the $s\times s$ dimensional, symmetric, and
positive matrices.
We denote the Gaussian measure on $\mathbb{R}^s$ with covariance $M(t)$ by
\begin{align}\label{eq:defofmuM}
\mu_{M(t)}(\d x)=\p_{M(t)}(x)\,\d x=\frac{e^{-\frac{1}{2}\langle x,M(t)^{-1}x\rangle}}{\sqrt{\text{det}(2\pi M(t))}} \,\d x.
\end{align}
Let $F\in C_b(\mathbb{R}^s,\mathbb{R})$ be a bounded and continuous function. 
We are interested in a bound for the expression 
\begin{align}\label{eq:defofpM}
 \left|\left.\frac{\d}{\d t} \int_{\mathbb{R}^s}  F(x) \,\mu_{M(t)}(\d
x)\right|_{t=0}\right|.
\end{align}
In principle this seems easy because the Gaussian integral acts as a heat
semi-group which is infinitely smoothing. However this is only true as long as
the eigenvalues do not approach zero (think of the delta distribution which is
a (degenerate) Gaussian measure). Therefore lower bounds on the eigenvalues of
$M(t)$ control the smoothing behaviour of the semigroup. Now we discuss the
necessary bounds in a bit more detail.
An elementary calculation, with the abbreviations 
 $M=M(0)$ and $\dot{M}=\dot{M}(0)$,
shows
\begin{align}\begin{split}\label{firstderivative}
\left|\left.\frac{\d}{\d t} \int_{\mathbb{R}^s}  F(x) \,\mu_{M(t)}(\d x)\right|_{t=0}\right|
&=\left|\frac{1}{2} \int_{\mathbb{R}^s}  F(x)\left(\langle x, M^{-1}\dot{M}M^{-1}x\rangle-\tr M^{-\frac{1}{2}}\dot{M}M^{-\frac{1}{2}}\right)
\,\mu_{M}(\d x)\right|.
\end{split}
\end{align}
The trace term arises as the derivative of the determinant.
To bound this expression one needs bounds on $M^{-1}$, i.e.,  lower bounds on the spectrum of $M$ are required.
With the help of the Cauchy-Schwarz inequality it can be shown that
\begin{align}\label{eq:finalderivbound}
\left|\left.\frac{\d}{\d t} \int_{\mathbb{R}^s}  F(x) \,\mu_{M(t)}(\d x)\right|_{t=0}\right|\leq \lVert
F\rVert_{L^2(\mathbb{R}^s,\mu_{M(t)})}\lVert M^{-\frac{1}{2}}\dot{M}M^{-\frac{1}{2}}\rVert_{\mathrm{HS}}.
\end{align}
To bound the right hand side of this equation we need lower bounds on the finite range decomposition and the derivatives of the kernels with
respect to $A$ have to decay faster than the kernels itself.
Denote $\dot{\ms{C}}_{k,A}=\frac{\d}{\d t} {\ms{C}}_{k,A+t\dot{A}}$ where
$\dot{A}$ is a fixed symmetric operator.
Then the Hilbert Schmidt norm from the right hand side of
\eqref{eq:finalderivbound} corresponds for scalar fields ($m=1$)
to the expression
\begin{align}\label{eq:keyCkexpression}
  \lVert
\ms{C}_{N,A}^{-\frac{1}{2}}\dot{\ms{C}}_{N,A}\ms{C}_{N,A}^{-\frac{1}{2}}\rVert_{
\mathrm{HS}}^2=
 \sum_{p\in \widehat{T}_N\setminus \{0\}}
\left(\frac{\widehat{\mc{\dot{C}}}_{N,A}(p)}{\widehat{\mc{C}}_{N,A}(p)}
\right)^2.
\end{align}
In other words this expression shows that the derivative of the expectation is not controlled
by the change of the covariance but rather by the relative change of the covariance. Moreover it
is more difficult to bound the Hilbert-Schmidt norm for bigger number of degrees of freedoms, i.e., increasing torus size.
We observe that since we have no  lower bound for all $\p\in \widehat{T}_N$ in Theorem
\ref{AKMFRD} we cannot bound the expression \eqref{eq:keyCkexpression}. Moreover this is not an issue of
missing bounds. This can be seen easily for the decomposition constructed in
\cite{adams2013finite}.
Their decomposition has the property that
$\ms{C}_{k,tA}=t^{-1}\ms{C}_{k,A}$.
This implies $\left.\frac{\d}{\d t}\ms{C}_{k,A-tA}\right|_{t=0}=\ms{C}_{k,A}$,
i.e., 
each summand in \eqref{eq:boundsketch} is 1. Hence the entire sum diverges like
$L^{Nd}-1$.
Therefore we have to modify the construction  such that we obtain
better lower bounds for the kernels while their derivatives continue to have good decay properties. 
With the decomposition from Theorem \ref{finalFRD} the bounds
\eqref{Ahatestimate} and \eqref{keyquotientbound} the expression
\eqref{eq:finalderivbound} can be bounded uniformly in $N$ for
$\tilde{n}-n>{d}/{2}$ in the scalar case as follows
 \begin{align}\begin{split}\label{eq:boundsketch}
 \lVert \ms{C}_{N,A}^{-\frac{1}{2}}\dot{\ms{C}}_{N,A}\ms{C}_{N,A}^{-\frac{1}{2}}\rVert_{\mathrm{HS}}^2=
 \sum_{p\in \widehat{T}_N\setminus \{0\}} \left(\frac{\widehat{\mc{\dot{C}}}_{N,A}(p)}{\widehat{\mc{C}}_{N,A}(p)}\right)^2&\leq C\sum_{j=0}^N
|\bd{A}_j|L^{2(N-j)(n-\tilde{n})}\\
& \leq C \sum_{j=0}^N L^{(N-j)(2(n-\tilde{n})+d)}\leq C.
\end{split}
 \end{align}
 This means that we can bound the derivative of expectation values of $F(\cdot+\xi_{N,A})$ with respect to $A$ uniformly in $N$ for the new
decomposition.
For $k<N$ one obtains similarly the bound 
\begin{align}\label{eq:generalCkheuristic}
 \lVert
\ms{C}_{k,A}^{-\frac{1}{2}}\dot{\ms{C}}_{k,A}\ms{C}_{k,A}^{-\frac{1}{2}}\rVert_{
\mathrm{HS}}^2\leq CL^{(N-k)d}.
\end{align}
This already indicates that for $k<N$ the derivative is bounded only for certain functionals $F$.
 For details and the general case $k<N$ and $m>1$ cf. Section \ref{secloss}.

Note that 
\eqref{firstderivative} can be also bounded using the following observation
\begin{align}
\left(\langle x, M^{-1}\dot{M}M^{-1}x\rangle-\tr M^{-\frac{1}{2}}\dot{M}M^{-\frac{1}{2}}\right) \,\mu_{M}(\d x)
=\sum_{i,j=1}^s (\dot{M}_{i,j}\partial_{x_i}\partial_{x_j}\p_{M}(x))\, \d x.
\end{align}
Integration by parts then implies
\begin{align}
\left|\left.\frac{\d}{\d t} \int_{\mathbb{R}^s}  F(x) \,\mu_{M(t)}(\d x)\right|_{t=0}\right|
=\left|\frac{1}{2}\int_{\mathbb{R}^s} \sum_{i,j=1}^s \dot{M}_{i,j}\left(\partial_{x_i} \partial_{x_j}F(x)\right) \,\mu_{M}(\d x)\right|
\end{align}
The bounds on $M^{-1}$ are no longer needed. Now, however, we  bound an integral over $F$ by an 
integral over the second derivative of $F$, i.e.,  we lose two orders of
regularity. 
This loss of regularity causes substantial difficulties in the renormalisation analysis in \cite{adams2016preprint} which can be avoided by using the decomposition from
Theorem \ref{finalFRD}. 
\section{Construction of the Finite Range Decomposition}\label{secconstruction}
The lower bound for the Fourier transform of the kernel of the finite range
decompositions $\mc{C}_k$ for $|p|\lesssim
L^{-k}$ allows one to construct
a new finite range decomposition which satisfies a global bound from below.
The key idea is to use suitable linear combinations 
of the original decomposition, i.e., we use the ansatz $\ms{D}_{k}=\sum_{j=1}^k \lambda_{k,j}\ms{C}_j$. By construction of the $\ms{C}_j$
the range of $\ms{D}_k$ is not larger than    $L^k/2$. 
The discrete derivatives of $\ms{D}_{k,A}$ shall be bounded as in \eqref{discretebounds}
for all $|\alpha|\leq n$ for some integer $n>0$.
Thus we need for $|\alpha|\leq n$ the estimate $\lambda_{k,j}|\nabla^\alpha\mc{C}_{j}(x)|\leq\lambda_{k,j}L^{-(j-1)(d-2+|\alpha|)}\leq
L^{-(k-1)(d-2+|\alpha|)}$
which is satisfied if $\lambda_{k,j}\leq L^{-(k-j)(d-2+n)}$. These bounds on $\lambda_{k,j}$ are the
largest possible (later we will add 1 in the exponent so that the sum over $j$ is still uniformly bounded).
Then for $p\approx L^{-j}$ with $j\leq k$ we find, using $\widehat{\mc{C}}_j(p)\gtrsim L^{2j}\approx |p|^{-2}$, a lower bound
\begin{align}
\widehat{\mc{D}}_k(p)\geq\lambda_{k,j}\widehat{\mc{C}}_j(p)\gtrsim L^{-(k-j)(d-2+n)}|p|^{-2}\approx 
|p|^{-2}(L^kp)^{-(d-2+n)}.
\end{align}
This decays much slower in Fourier space than the decomposition from Theorem \ref{AKMFRD} and therefore it is  helpful
to bound the expression in  \eqref{eq:finalderivbound}.
The construction above yields a decomposition with good lower and upper global bounds on the Fourier modes of the finite range decomposition. The
following proposition states the precise result.
\begin{proposition}\label{lowerboundFRD}
 Let $n>0$ be an integer. Then the family $\ms{C}_A$ of operators with $A\in \mc{L}(\mc{G},\omega_0,\Omega_0)$
 has
 a finite range decomposition into operators $\ms{D}_{A,k}$ such that $\ms{C}_A=\sum_{k=0}^{N+1} \ms{D}_{A,k}$ and
  \begin{align}
\mc{D}_{A,k}(x)=M_k\,\text{ if }\, |x|_\infty\geq\frac{L^k}{2}
 \end{align}
 where $M_k\leq 0$ is a constant matrix.
Furthermore for any multiindex $\alpha$ with $|\alpha|\leq n$ and the  constants
$C(\alpha,\ell)$ 
 from Theorem \ref{AKMFRD} the discrete derivative in $x$ and the directional
derivative in $A$   satisfy
\begin{align}\label{eqdirectderivative}
\sup_{\lVert\dot{A}\rVert\leq 1} |\nabla^\alpha
D_A^\ell\mc{D}_{A,k}(x)(\dot{A},\cdots,\dot{A})|\leq
2C(\alpha,d)L^{-(k-1)(d-2+|\alpha|)}.
\end{align}
Moreover we also have a lower bound on $\ms{D}_{A,k}$ in Fourier space
\begin{align}\label{lowerfourierbounds}
 \widehat{\mc{D}}_{A,k}(p)&\geq
 \begin{cases}
  cL^{2j}L^{(k-j)(-d+1-n)}&\text{ for } p\in \bd{A}_j\text{ and }  j< k\\
 cL^{2k}&\text{ for } p\in \bd{A}_j\text{ and }
j\geq k.
\end{cases}
\end{align}

The upper bound in Fourier space reads
\begin{align}\label{upperfourierbounds}
\lVert D^\ell_A\widehat{\mc{D}}_{A,k}(p)\rVert\leq 
\begin{cases}
CL^{2(n+d)+1}L^{2j}L^{(k-j)(-d+1-n)}
\;&\text{for } p\in \bd{A}_j\text{ and }  j< k\\
CL^{2k}\;&\text{
for}\; p\in \bd{A}_j\text{ and } j\geq k.
\end{cases}
\end{align}
In particular the quotient between the lower bound \eqref{lowerfourierbounds} and the upper bound \eqref{upperfourierbounds} 
is bounded by a constant for all $p\in \widehat{T}_N$ and $A, A'\in \mc{L}(\mc{G},\omega_0,\Omega_0)$, i.e., 
\begin{align}\label{thekeyquotientbound}
KL^{2(d+n)+1}\widehat{\mc{D}}_{A,k}(p)\geq \lVert D_{A'}^\ell
\widehat{\mc{D}}_{A',k}(p)\rVert
\end{align}
for some constants $K=K(n,\ell)$. 
\end{proposition}
\begin{remark}
\begin{enumerate}
\item The finite range decomposition $\ms{D}_{A,k}$ has the property
\begin{align*}
 \ms{D}_{A,k+1}\geq L^{-d+1-n}\ms{D}_{A,k}
\end{align*}
which can be seen easily from the construction below.
One easily sees that $L^{-d+1-n}$ can be
replaced by $\eta L^{-d+2-n}$ for any $\eta<1$.
 This bound seems to be optimal
  under the condition that the discrete derivatives up to
order $n$ are bounded as in \eqref{eqdirectderivative} because the bound for
$\lVert \nabla^\alpha \mc{D}_{k,A}\rVert_\infty$ strengthens by a factor of $L^{d-2+n}$
in each step if $|\alpha|=n$. 
Lower bounds of this type might be useful for a new approach to the definition
of the  norms for the renormalisation group approach.
\item The construction is slightly more flexible when we start from the
continuous decomposition in \cite{MR3129804} (cf. Appendix A
). Then we
define 
\begin{align}
 \ms{D}_{k,A}=\int_{\mathbb{R}_+} \Phi_k(t)\cdot t W_t(\ms{A})\, \d t
\end{align}
where $\Phi_k:\mathbb{R_+}\rightarrow \mathbb{R}_+$ is a family of positive functions that are a decomposition of unity with
$\textrm{supp}(\Phi_k)\subset (0,L^k/(2R))$ and it behaves as $\Phi_k(t)\approx
(tL^{-k})^{d-1+n}$ for $t<L^k/(2R)$.
Since this is more technical we used to the simpler construction below
based on the discrete decomposition in Theorem \ref{AKMFRD}. 
\item
 Note that the only relevant difference between the result of this Proposition and Theorem \ref{finalFRD} is the decay of
$D^\ell_A\widehat{\mc{D}}_{A,k}(p)$ for $\ell>0$ which is here given by
$L^{(k-j)(-d+1-n)}$. In Theorem \ref{finalFRD} the decay is improved to $L^{(k-j)(-d+1-\tilde{n})}$ where $\tilde{n}>n$ is any integer. However,
this small change helps a lot with the sum in equation \eqref{eq:boundsketch} which cannot be bounded by using Proposition \ref{lowerboundFRD}.
\end{enumerate}
\end{remark}
\begin{proof}
 Let $\mathcal{C}_k$ be a finite range decomposition as in Theorem \ref{AKMFRD}. Define 
 \begin{align}
  \ms{D}_k=\sum_{j=1}^k\lambda_{k,j}\ms{C}_j
 \end{align}
where the coefficients $\lambda_{k,j}$ are given by $\lambda_{k,j}=L^{(k-j)(-d+2-n-1)}$ for $j<k$ and $\lambda_{k,k}$ is defined implicitly by
$\sum_{l=k}^{N+1} \lambda_{l,k}=1$. This condition implies $\sum_{k=1}^{N+1}\ms{D}_k=\ms{C}$. Since $\lambda_{l,k}$ is a geometric series in $l$
for $l>k$ we
find that $1\geq\lambda_{k,k}>1/2$. The operators $\ms{D}_k$ clearly have the correct range. The discrete derivatives can be estimated easily for
$|\alpha|\leq n$ using \eqref{discretebounds}
\begin{align}\begin{split}
 |\nabla^\alpha D_A^\ell \mc{D}_{A,k}(x)(\dot{A},\ldots,\dot{A})|&\leq \sum_{j=1}^{k}\lambda_{k,j}|\nabla^\alpha D_A^\ell
\mc{C}_{A,j}(x)(\dot{A},\ldots,\dot{A})|\\
& \leq \sum_{j=1}^{k}
C(\alpha,\ell)L^{(k-j)(-d+2-n-1)-(j-1)(d-2+|\alpha|)}\\
& \leq
C(\alpha,\ell)L^{-(k-1)(d-2+|\alpha|)}
\sum_{j=1}^{k} L^{(j-k)(n-|\alpha|+1)}\\
&\leq 2C(\alpha,\ell)L^{-(k-1)(d-2+|\alpha|)}.
\end{split}\end{align}
 In the last step we used $n-|\alpha|\geq 0$ and we estimated the geometric series by 2. 
It remains to prove the  bounds in Fourier space.
For $j\geq k$ and $p\in \bd{A}_j$ we use \eqref{keylowerbound} and
$\lambda_{k,k}\geq \frac{1}{2}$
\begin{align}
 \widehat{\mc{D}}_k(p)\geq \frac{1}{2}\widehat{\mc{C}}_k(p)\geq
\frac{c}{2}L^{2k}.
\end{align}
For $1\leq j< k$ we employ again \eqref{keylowerbound} on the  $\ms{C}_{j}$
summand of $\ms{D}_k$
\begin{align}\begin{split}
 \widehat{\mc{D}}_k(p)&\geq \lambda_{k,j}\widehat{\mc{C}}_{j}(p)\geq
cL^{(k-j)(-d+2-n-1)}L^{2j}
 \geq cL^{(k-j)(-d+1-n)}L^{2j}.\end{split}
\end{align}
Finally, for $j= 0$ equation \eqref{keylowerbound} applied to the $\ms{C}_1$ term yields 
\begin{align}\begin{split}
 \widehat{\mc{D}}_k(p)&\geq \lambda_{k,1}\hat{\mc{C}}_{1}(p)\geq
L^{(k-1)(-d+1-n)}c|p|^{-2}
 \geq c L^{k(-d+1-n)} L^{d+n-1}L^{-2}\geq  c L^{(k-j)(-d+1-n)}L^{2j}.\\
\end{split}\end{align}
The proof of the upper bound \eqref{upperfourierbounds}  is straightforward but technical.
For $j\geq k$ and $p\in \bs{A}_j$ the bound is immediate from the second estimate of
\eqref{fourierbounds}
 because 
  \begin{align}
\lVert\widehat{\mc{D}}_k(p)\rVert \leq \sum_{k'=1}^{k}\lambda_{k,k'}\lVert \widehat{\mc{C}}_{k'}(p)\rVert
\leq  C\sum_{k'=1}^{k}L^{2k'}\lambda_{k,k'}\leq
2CL^{2k}.
\end{align}
On the other hand for $j< k$ and $p\in \bs{A}_j$ we find with \eqref{fourierbounds}
for any $\bar{n}>0$
\begin{align}\label{frd2est1}\begin{split}
&\lVert \widehat{\mc{D}}_k(p)\rVert\leq \sum_{k'=1}^{k}\lambda_{k,k'}\lVert \widehat{\mc{C}}_{k'}(p)\rVert\leq\\
&\leq \sum_{k'=1}^{j}
L^{(k-k')(-d+1-n)}CL^{2k'}
+\sum_{k'=j+1}^{k}L^{(k-k')(-d+1-n)}C_{\bar{n}}|p|^{-2} (|p|L^{
k'-1})^{-\bar{n}}.
\end{split}\end{align}
The first summand is a geometric sum bounded by twice the largest term
\begin{align}\label{frdest2}
\sum_{k'=1}^{j}
L^{(k-k')(-d+1-n)}CL^{2k'}\leq 2C L^{(k-j)(-d+1-n)}L^{2j}.
\end{align}
The second summand in \eqref{frd2est1} is also a geometric series and for
$\bar{n}=d+n $ it can be bounded similarly
\begin{align}\begin{split}\label{frdest3}
\sum_{k'=j+1}^{k}L^{(k-k')(-d+1-n)}C |p|^{-2}(|p|L^{
k'-1})^{-\bar{n}}
&\leq
CL^{2(j+1)} 
\sum_{k'=j+1}^{k}L^{(k-k')(-d+1-n)}(L^{
k'-j-2})^{-(d+n)}\\
&\leq 
C L^{2(d+n)+1}L^{2j}L^{(k-j)(-d+1-n)}.
\end{split}\end{align}
Now the estimates \eqref{frdest2} and \eqref{frdest3} plugged in
\eqref{frd2est1} imply \eqref{upperfourierbounds} for $\ell=0$.
This ends the proof for $\ell=0$. For $\ell>0$  only the constants in the upper
bound change.
\end{proof}
So far we have constructed a finite range decomposition with $n$ controlled discrete derivatives and matching lower and upper bounds
on the Fourier coefficients. Finally we want to prove Theorem \ref{finalFRD}
where we stated the existence of a finite range decomposition for  $\ms{C}_A$ 
where the derivatives with respect to $A$ of the kernel decay better in Fourier
space than the kernel itself. The key idea is to start with a decomposition as
in Proposition \ref{lowerboundFRD} 
with many controlled derivatives, then subtract something constant such that the
decomposition remains positive. Finally we add the subtracted part in a way that
we get strong lower bounds. 
Then the derivatives with respect to $A$ only hit the fast decaying first term.
The main problem
is to ensure that the operators remain positive.
In the scalar case the fact that all operators are simultaneously diagonalised by the Fourier transform simplifies the 
analysis slightly and this would allow one to obtain slightly stronger results than the one stated in Theorem \ref{finalFRD}.

For the construction we have to fix an operator where for simplicity we choose
$-\Delta$, where $\Delta=-\nabla^\ast \nabla$ denotes the lattice Laplacian.
The linear map $A:\mc{G}\rightarrow\mc{G}$ corresponding
to the discrete Laplacian acting on vectors in $\mathbb{R}^m$ has the matrix elements
$A_{\alpha\alpha}=\mathrm{id}_{m\times m}$ for $|\alpha|=1$ and 0 otherwise. 
With a slight abuse of notation we denote finite range decompositions of the
Laplacian by $\ms{D}_{\Delta,k}$.
\begin{proof}[Proof of Theorem \ref{finalFRD}]
 Let $\ms{D}_{A,k}^{\tilde{n}}$ and $\ms{D}_{-\Delta,k}^{\tilde{n}}$ be finite
range decompositions as constructed in Proposition
 \ref{lowerboundFRD} with $\tilde{n}$ controlled derivatives.
Let $\ms{D}_{-\Delta,k}^n$ be a finite range decomposition of
$\ms{C}_{-\Delta}$ as constructed in Proposition \ref{lowerboundFRD}
with derivatives up to order $n$ bounded.
Then we define
\begin{align}
 \ms{C}_{A,k}=\ms{D}_{A,k}^{\tilde{n}}-
 \frac{L^{-2(d+\tilde{n})-1}}{K(\tilde{n},\ell=0)}
 \ms{D}^{\tilde{n}}_{-\Delta,k}
 +\frac{L^{-2(d+\tilde{n})-1}}{K(\tilde{n},\ell=0)}
 \ms{D}_{-\Delta, k}^n
\end{align}
where $K\geq 1$ is the constant from the inequality \eqref{thekeyquotientbound}.
Clearly we have
\begin{align}
 \ms{C}_{A}=\sum_{k=1}^{N+1} \ms{C}_{A,k},
\end{align}
this decomposition has the correct finite range and it is translation invariant.
Bounds on the discrete derivatives already hold for each term separately by \eqref{eqdirectderivative}. The same is true
for the upper bounds in Fourier space. Moreover the stronger bound for the derivatives 
with respect to $A$ follows from the fact that only the first term depends on $A$ and
the bounds given in equation \eqref{upperfourierbounds} in Proposition \ref{lowerboundFRD}.
It remains to prove the lower bounds in Fourier space which also imply positivity of the decomposition.
The third term satisfies the lower bound so it remains to proof that the first two terms together are positive.
But this is a consequence of  \eqref{thekeyquotientbound} which in particular gives
\begin{align}
\widehat{\mc{D}}_{A,k}^{\tilde{n}}(p)\geq 
\frac{L^{-2(d+\tilde{n})+1}}{K(\tilde{n},\ell=0)}\lVert
\widehat{\mc{D}}_{-\Delta,k}^{\tilde{n}}(p)\rVert.
\end{align}
\end{proof}

\section{Smoothness of the Renormalisation Map}\label{secloss}
The goal of this section is to prove the following proposition stating the smoothness of the
renormalisation map based on the finite range decomposition from Theorem \ref{finalFRD}.
This is similar to the motivation given in Section \ref{secconstruction}. 
We first discuss the simplest smoothness statement as an illustration without technical problems. Then, in Proposition
\ref{prop:finalsmoothness}, we prove the more general smoothness statement that allows to
avoid the loss of regularity in \cite{adams2016preprint}.
Both propositions rely on an interesting localisation property of finite range
decompositions discussed in Lemma \ref{lemmaloc}.
\begin{proposition}\label{proprenormmap}
Let $B\subset T_N$ be a cube of side length $L^k$. 
Let $F:\mc{X}_N\rightarrow \mathbb{R}$ be a bounded functional that is measurable with respect
to the $\sigma$-algebra generated by $\{ \p(x)|\, x\in B\}$, i.e., $F$ depends only on the restriction $\left.\p\right|_B$. 
Then the following bound holds
\begin{align}\label{eq:special}
\left|\partial_A \int_{\mc{X}_N} F(\p+\psi)\,\mu_{\ms{C}_{A,k+1}}(\d \p)\right|\leq C\lVert F\rVert_{L^2(\mc{X}_N,\mu_{\ms{C}_{A,k+1}})}
\end{align} 
where $\ms{C}_{A,k+1}$ is a finite range decomposition as in Theorem \ref{finalFRD}
with $\tilde{n}-n>d/2$ and $C$ is a constant that does not depend on $N$ or $k$
but in contrast to the previous sections it does depend on $L$.
\end{proposition}
\begin{remark}
 The condition that $F$ depends on the values of the field in a cube of
side length $L^k$ appears naturally in the renormalisation group analysis, cf., 
e.g.,  \cite{MR2523458}.
\end{remark}
As stated in Proposition \ref{proprenormmap} we show a bound for derivatives of the renormalisation map which does not depend on $N$.
In principle this could be difficult since the dimension of the space $\mc{X}_N$ over which we integrate increases with $N$.
We have seen in equation \eqref{eq:boundsketch} in Section \ref{secsetting} that the number of terms appearing in the 
 bounds for the
derivatives with respect to the covariance is proportional to the dimension of
the space $\mc{X}_N$ and therefore that the naive estimate is only sufficient to bound the derivative for $k=N$. For $k<N$ \eqref{eq:generalCkheuristic}
suggests that there is no uniform in $N$ bound for general functionals $F$. This means we  have to exploit
the special structure of the functionals.

The following heuristics suggests that the specific structure of the functionals is indeed sufficient to bound the 
derivative uniformly in $N$.
We integrate in the $k$-th integration step with respect to a measure of range $L^k$
which satisfies bounds uniformly in $N$ and the functionals
$F_k(\p)$ that appear are local in the sense that they only depend on the values of $\p$ on $B$.
Hence in some sense the size of the torus should not be seen by this
integration. 
 This idea can be made rigorous in the sense that locally the distribution on a block $B$ of a field $\xi$ with 
distribution $\mu_{\ms{C}_{A,k}}$ is the same as the distribution of a field $\tilde{\xi}$ defined on a torus
of size comparable to $B$.

A bit more quantitatively we motivate the $\mc{O}(1)$ bound as follows. From
\eqref{eq:generalCkheuristic} we find the bound 
\begin{align}\label{eq:heuristicbd}
 \left|\partial_A \int_{\mc{X}_N} F(\p+\psi)\,\mu_{\ms{C}_{A,k+1}}(\d
\p)\right|\leq C\sqrt{L^{(N-k)d}}\lVert F\rVert_{L^2(\mc{X}_N,\mu_{A,k+1})}
\end{align}
for general functionals $F$.
If $F$ depends only on the values of $\p$ in a block of side-length $L^k$ it
depends only on the fraction $L^{(k-N)d}$ of the degrees of freedom. 
We expect that each degree of freedom contributes equally which suggests the
bound stated in \eqref{eq:special}.

Let us sketch a proof of Proposition \ref{proprenormmap} that makes the previous consideration precise. Below, we also give a longer proof based on Lemma \ref{lemmaloc} because the second proof can be used to establish Theorem \ref{prop:finalsmoothness}.
Consider the set $T$ of translation operators given by 
\begin{align}
T=\{\tau_a: a=(a_1,\ldots ,a_n)\in (3L^{k+1}\mathbb{Z})^d,\,
 0\leq a_i\leq L^N-3L^{k+1}\}.
 \end{align} Then there is a constant such that $|T|\geq cL^{(N-k-1)d}$.
 Note that for the translation $\tau_a$ the random variable 
 $F(\tau_a\p)$ only depends on the values of $\p$ in $B+a$. The definition of $T$ implies
 that for $\tau_{a_1}, \tau_{a_2}\in T$ with $\tau_{a_1}\neq \tau_{a_2}$ the sets $B+a_1$ and $B+a_2$
 have distance at least $L^{k+1}$ which implies that the random variables
 $F(\tau_{a_1}\p)$ and $F(\tau_{a_2}\p)$ are independent. 
Then we can estimate using translation invariance of $\mu_{A,k+1}$, independence of $\tau_a F$, and \eqref{eq:heuristicbd}
\begin{align}
\begin{split}
\left|\partial_A \int_{\mc{X}_N} F(\p)\,\mu_{A,k+1}(\d \p)\right|
&= \frac{1}{|T|}
\left|\partial_A \sum_{\tau\in T} \int_{\mc{X}_N} F(\tau\p)\,\mu_{A,k+1}(\d \p)\right|
\\
&\leq 
\frac{C \sqrt{L^{(N-k-1)d}}}{|T|} \left\lVert \sum_{\tau\in T} F(\tau\p)\right\rVert_{L^2(\mc{X}_N,\mu_{A,k+1}} 
\\
&\leq \frac{C \sqrt{L^{(N-k-1)d}}}{|T|}  \left(\sum_{\tau\in T} \int_{\mc{X}_N} F^2(\tau\p)\,\mu_{A,k+1}\right)^{\frac12}
\\ 
&\leq 
C\sqrt{\frac{L^{(N-k-1)d}}{|T|}}   \lVert F\rVert_{L^2(\mc{X}_N,\mu_{A,k+1})}
\\
&\leq 
C   \lVert F\rVert_{L^2(\mc{X}_N,\mu_{A,k+1})}.
\end{split}
\end{align}

We introduce some notation necessary for the next lemma. 
Let $k\leq\overline{N}\leq N$ be  positive integers.
We denote by $\pi=\pi_{N,\overline{N}}:T_N\rightarrow T_{\overline{N}}$ the
projection and group homomorphism of discrete tori.
Recall that $\mc{V}_N$ was defined in \eqref{eq:defofVN} and denotes the set of
fields on $T_N$.
Let $\tau:\mc{V}_{\overline{N}}\rightarrow \mc{V}_N$ with $\tau=\pi^{\ast}$ be
the pull-back of fields, i.e.,
for $\p\in\mc{V}_{\overline{N}}$ we define $(\tau\p)(x)=\p(\pi x)$. Clearly
$\tau\p$ has average zero if $\p$ has average zero hence $\tau$ also maps the
subspace $\mc{X}_{\overline{N}}$ to $\mc{X}_N$. In other words, identifying
functions
on $T_N$ with periodic functions on $\mathbb{Z}^d$, the function $\tau\p$ is
just the  $(L^{\overline{N}}\mathbb{Z})^d$ periodic function $\p$ understood as
a $(L^{N}\mathbb{Z})^d$ periodic function.
With this notation we can state the following lemma.
Here $\ms{C}=\sum_{k=0}^{N+1}\ms{C}_k$ 
denotes a finite range decomposition such that $\ms{C}_k$ is a
non-negative translation invariant, positive operator on $\ms{V}_N$ (see Remark \ref{rem:XvsV} below) with kernels $\mc{C}_k$
satisfying $\mc{C}_k(x)=-M\leq 0$ for $\d_\infty(x,0)\geq {L^k}/{2}$ where $M$
is a positive semi-definite symmetric matrix.
\begin{lemma}\label{lemmaloc}
Let $X\subset T_N$  and $D=\mathrm{diam}(X)=\sup_{x,y\in X} d_\infty(x,y)$.
Choose
$\overline{N}\in\mathbb{N}$ such that $L^{\overline{N}}> 2D\geq
L^{\overline{N}-1}$ and assume $k\leq \overline{N}\leq N$.
Define a  Gaussian measure $\nu_{k,\overline{N}}$ on
$\mc{V}_{\overline{N}}$ by its covariance operator $\ms{D}_{k,\overline{N}}$
given by the kernel
$\mc{D}_{k,\overline{N}}:T_{\overline{N}}\rightarrow \mathbb{R}^{m\times
m}_{\mathrm{sym}}$
\begin{align}
\mc{D}_{k,\overline{N}}(x)=(L^{(N-\overline{N})d}-1)M+\frac{1}{L^{\overline{N}d}}\sum_{p\in
\widehat{T}_{\overline{N}}} e^{ipx}
\widehat{\mc{C}}_k(p)
\end{align}
where $\widehat{\mc{C}}_k(p)$ are the Fourier coefficients of the kernel
$\mc{C}_k$ of the covariance operator $\ms{C}_k$. They are well defined because
$\widehat{T}_{\overline{N}}\subset \widehat{T}_N$.
Let $F_X:\mc{V}_N\rightarrow \mathbb{R}$ 
be a measurable functional that only depends on $\{\p(x)|\, x\in X\}$, i.e., $F$ is measurable with respect to the $\sigma$-algebra generated by
$\{\p(x)|\, x\in X\}$. Then the following identity holds
\begin{align}\label{eq:lemmaloc}
\int_{\mc{V}_N} F_X(\xi)\,\mu_k(\d \xi)=\int_{\mc{V}_{\overline{N}}}
F_X(\tau\psi)\,\nu_{k,\overline{N}}(\d\psi).
\end{align}
\end{lemma}
\begin{remark}\label{rem:XvsV}
Note that the measure $\mu_k$ appearing on the left hand side of equation \eqref{eq:lemmaloc} was defined  on $\mc{X}_N \subset \mc{V}_N$. In Remark \ref{remark:VnvsXn}
we discussed that $\mu_k$ agrees with a degenerate Gaussian measure on $\mc{V}_N$ which implies that the left hand side is well defined.
\end{remark}
\begin{proof} 
The proof of the lemma is divided in two steps. 

Step 1: First we show that we can find another Gaussian measure on $\mc{V}_N$
such that
the local covariance structure is the same but all Fourier coefficients of the
kernel except the ones that satisfy $p\in\widehat{T}_{\overline{N}}\subset
\widehat{T}_N$ vanish (see \eqref{eq:covident} below).
Define the kernel of an operator $\mc{C}_{k,\overline{N}}:T_N\rightarrow
\mathbb{R}^{m\times m}_{\mathrm{sym}}$ by
\begin{align}
\mc{C}_{k,\overline{N}}(x)=\frac{1}{L^{d\overline{N}}}\sum_{p\in
\widehat{T}_{\overline{N}}} e^{ipx}\widehat{\mc{C}}_k(p)+\lambda M
\end{align}
where here and in the following we write $\lambda=L^{(N-\overline{N})d}-1$.
Clearly the Fourier modes of $\mc{C}_{k,\overline{N}}$ are given by
$\widehat{\mc{C}}_{k,\overline{N}}(p)=L^{(N-\overline{N})d}\widehat{\mc{C}}_k(p)$
for $p\in \widehat{T}_{\overline{N}}\setminus\{0\}$,
$\widehat{\mc{C}}_{k,\overline{N}}(0)=L^{(N-\overline{N})d}\widehat{\mc{C}}_k(p)+L^{Nd}\lambda M$, and 0 otherwise.
By assumption $\ms{C}_k$ is a non-negative operator on $\mc{V}_N$, i.e.,  all
Fourier coefficients 
are non-negative and $M$ is positive by
assumption, hence $\ms{C}_{k,\overline{N}}$ also is a non-negative operator.
Therefore it defines a (highly degenerate) Gaussian measure
$\mu_{k,\overline{N}}$ on $\mc{X}_N$ with covariance operator
$\ms{C}_{k,\overline{N}}$.

We are interested in this measure because of the following remarkable property:
For $x\in T_N$ with $\d_\infty(x,0)\leq {(L^{\overline{N}}-1)}/{2}$ we have
$\mc{C}_{k,\overline{N}}(x)=\mc{C}_k(x)$, i.e.,  locally the measures $\mu_k$
and $\mu_{k,\overline{N}}$ have the same covariance structure 
Let us prove this property. First we observe	that for
$p\in \widehat{T}_{\overline{N}}$ the exponential $f_p$ is a well defined
function on $T_{\overline{N}}$ and
with slight abuse of notation (identifying the exponentials on $T_N$ and
$T_{\overline{N}}$)
it satisfies $f_p(x)=f_p(\pi x)$ for $x\in T_N$. Hence we find for $x,y\in T_N$
\begin{align}
\frac{1}{L^{\overline{N}d}}\sum_{p\in \widehat{T}_{\overline{N}}} e^{ip(x-y)}=
\frac{1}{L^{\overline{N}d}}\sum_{p\in \widehat{T}_{\overline{N}}} e^{ip\cdot
\pi(x-y)}=\begin{cases}
1 \; \text{if } \; \pi(x)=\pi(y)\\
0\; \text{else.}
\end{cases}
\end{align}
Using this, we calculate
\begin{align}
\begin{split}
\mc{C}_{k,\overline{N}}(x)&=\frac{1}{L^{\overline{N}d}}\sum_{p\in
\widehat{T}_{\overline{N}}} e^{ipx}\widehat{\mc{C}}_k(p)+\lambda M\\
&=\frac{1}{L^{\overline{N}d}}\sum_{p\in \widehat{T}_{\overline{N}}}
e^{ipx}\left(\sum_{y\in T_N}e^{-ipy}\mc{C}_k(y)\right)+\lambda M\\
&=\sum_{y\in T_N}\mc{C}_k(y)\left(\frac{1}{L^{\overline{N}d}}\sum_{p\in
\widehat{T}_{\overline{N}}} e^{ip(x-y)}\right)+\lambda M\\
&=\sum_{\substack{y\in T_N \\ \pi(y)=\pi(x)}}\mc{C}_k(y)+\lambda M.
\end{split}
\end{align}
In the first step we used the definition of the kernel, in the second
we used the definition of the Fourier transform and
the third step interchanged the order of summation.
Now for a given point $x\in T_N$ there is exactly one $y_0\in T_N$ such that 
$\pi(x)=\pi(y_0)$ and $\d_\infty(y_0,0)\leq \frac{L^{\overline{N}}-1}{2}$. If
$\d_\infty(x,0)\leq \frac{L^{\overline{N}}-1}{2}$ we have $x=y_0$.
Moreover, for $\d_\infty(y,0)> \frac{L^{\overline{N}}-1}{2}\geq \frac{L^k-1}{2}$
we have
$\mc{C}_k(y)=-M$ by assumption.
Hence we have 
\begin{align}\label{eq:covident}
 \mc{C}_{k,\overline{N}}(x)=\mc{C}_k(x)-(L^{(N-\bar{N})d}-1)M+\lambda
M=\mc{C}_k(x)
\end{align}
for $\d_\infty(x,0)\leq
\frac{L^{\overline{N}}-1}{2}$
 as claimed.
Actually we have even shown that $\mc{C}_{k,\overline{N}}$ is the
$(L^{\overline{N}}\mathbb{Z})^d$ periodic extension
of $\left.\mc{C}_k\right|_{\left[-\frac{L^{\overline{N}}-1}{2},\ldots,\frac{L^{\overline{N}}-1}{2}\right]^d}$.

Next we claim that if $\xi$ is distributed according to $\mu_k$ and $\p$ is distributed according to 
$\mu_{k,\overline{N}}$ then
$\left.\xi\right|_{X}\stackrel{\text{Law}}{=}\left.\p\right|_{X}$.
First we note that since the distribution of $\xi$ and $\p$ is Gaussian with mean zero the same holds
for the restrictions $\left.\xi\right|_{X}$ and $\left.\p\right|_{X}$. Then it is
enough to prove that all covariances agree because they determine the law.
By the assumption $x,y\in X\Rightarrow \d_\infty(x,y)\leq D \leq
\frac{L^{\overline{N}}-1}{2}$. Hence
$\mc{C}_k(x-y)=\mc{C}_{k,\overline{N}}(x-y)$ and therefore
\begin{align}
\mathbb{E}\left(\xi^i(x)\xi^j(y)\right)=\mc{C}^{ij}_k(x-y)=\mc{C}^{ij}_{k,
\overline{N}}(x-y)=\mathbb{E}\left(\p^i(x)\p^j(y)\right).
\end{align}
By assumption there exists a functional $\tilde{F}_X:(\mathbb{R}^m)^X\rightarrow \mathbb{R}$ such that
$F_X(\xi)=\tilde{F}_X(\left.\xi\right|_{X})$. Hence we find
\begin{align}\label{endstep1}\begin{split}
\int_{\mc{V}_N} F_X(\xi)\,\mu_k(\d \xi)=
\int_{\mc{V}_N} \tilde{F}_X(\left.\xi\right|_{X})\,\mu_k(\d \xi)&=
\int_{\mc{V}_N} \tilde{F}_X(\left.\p\right|_{X})\,\mu_{k,\overline{N}}(\d \p)\\
&=\int_{\mc{V}_N} F_X(\p)\,\mu_{k,\overline{N}}(\d \p).
\end{split}
\end{align}

Step 2: In the second step we show that the measure $\mu_{k,\overline{N}}$ is the push-forward of 
$\nu_{k,\overline{N}}$ along
$\tau$, i.e.,  $\mu_{k,\overline{N}}=\tau_{\ast}\nu_{k,\overline{N}}$.
 The
measure
$\nu_{k,\overline{N}}$ was defined in the statement of the lemma by the kernel
$\mc{D}_{k,\overline{N}}:T_{\overline{N}}\rightarrow \mathbb{R}$
\begin{align}
\mc{D}_{k,\overline{N}}(x)=\frac{1}{L^{\overline{N}d}}\sum_{p\in
\widehat{T}_{\overline{N}}} e^{ipx} \widehat{\mc{C}}_k(p)+\lambda M.
\end{align}
From this equation we can 
similarly to Step 1 extract the Fourier decomposition of this operator
and see that this expression defines a non-negative operator and therefore the covariance of a Gaussian measure.
Note that for $x\in T_N$
\begin{align}\label{cisequaltod}
\mc{D}_{k,\overline{N}}(\pi(x))=\mc{C}_{k,\overline{N}}(x)
\end{align}
which is again a consequence of $f_p(\pi x)=f_p(x)$ for
$p\in\widehat{T}_{\overline{N}}$. In other words
the kernel $\mc{C}_{k,\overline{N}}$ is already $(L^{\overline{N}}\mathbb{Z})^d$
periodic and hence also defines a function
$T_{\overline{N}}\rightarrow \mathbb{R}$ which we call
$\mc{D}_{k,\overline{N}}$. The previous
definition has the advantage that it makes clear that this kernel defines a Gaussian measure. 
The proof of $\mu_{k,\overline{N}}=\tau_{\ast}\nu_{k,\overline{N}}$ is standard.
We prove that the characteristic functions for both measures agree. Let
$v:T_N\rightarrow \mathbb{R}^m$ be a field. We have to show
\begin{align}
\int_{\mc{V}_N} e^{i\langle v,\p\rangle}\,\mu_{k,\overline{N}}(\d \p)=
\int_{\mc{V}_N} e^{i\langle v,\p\rangle}\,\tau_{\ast}\nu_{k,\overline{N}}(\d
\p).
\end{align}
The left hand side is the characteristic function of a Gaussian measure given by
\begin{align}\label{characteristicLHS}
\int_{\mc{V}_N} e^{i\langle v,\p\rangle}\,\mu_{k,\overline{N}}(\d
\p)=\exp\left(-\frac{\langle v,\ms{C}_{k,\overline{N}}v\rangle}{2}\right)
\end{align}
as completion of the square shows.
The right hand side is slightly more complicated. By the transformation formula we find
\begin{align}\begin{split}\label{characteristicRHS}
\int_{\mc{V}_N} e^{i\langle v,\p\rangle}\,\tau_{\ast}\nu_{k,\overline{N}}(\d \p)
&=\int_{\mc{V}_{\overline{N}}} e^{i\langle
v,\tau\psi\rangle}\,\nu_{k,\overline{N}}(\d \psi)\\
&=\int_{\mc{V}_{\overline{N}}} e^{i\langle \tau^{\ast}
v,\psi\rangle}\,\nu_{k,\overline{N}}(\d \psi)\\
&= e^{-\frac{1}{2}\langle \tau^{\ast}
v,\ms{D}_{k,\overline{N}}\tau^{\ast}v\rangle}.
\end{split}\end{align}
Here $\tau^{\ast}:\mc{V}_{N}\rightarrow \mc{V}_{\overline{N}}$ is the adjoint
 of $\tau:\mc{V}_{\overline{N}}\rightarrow \mc{V}_N$ with respect to the
standard scalar product on both spaces, i.e., $\tau^\ast$ is
characterised by $\langle \p,\tau\xi\rangle=\langle \tau^{\ast} \p,\xi\rangle $
for $\p\in\mc{V}_N$ and
$\xi\in\mc{V}_{\overline{N}}$.
It is easy to see that $\tau^{\ast}v(\overline{x})=\sum_{{x}\in
T_N:\pi(x)=\overline{x}} v(x)$.
Then we find
\begin{align}\begin{split}
\langle \tau^{\ast} v,\ms{D}_{k,\overline{N}}\tau^{\ast}v\rangle
&= \sum_{\overline{x},\overline{y}\in T_{\overline{N}}}
\tau^{\ast}v(\overline{x})\mc{D}_{k,\overline{N}}(\overline{x}-\overline{y}
)\tau^{ \ast} v(\overline{y })\\
&=\sum_{\overline{x},\overline{y}\in T_{\overline{N}}}
\sum_{\substack{x\in T_N \\ \pi(x)=\overline{x}}} \sum_{\substack{ y\in T_N\\
\pi(y)=\overline{y}}}
 v(x)\mc{D}_{k,\overline{N}}(\pi(x-y))v(y)\\
&=\sum_{x,y \in T_N}  v(x)\mc{C}_{k,\overline{N}}(x-y)v(y)\\
&=\langle v,\ms{C}_{k,\overline{N}} v\rangle.
\end{split}
\end{align}
Together with the equations \eqref{characteristicLHS} and \eqref{characteristicRHS} this shows the claim.

Conclusion: From equation \eqref{endstep1} and Step 2 we conclude
\begin{align}
\int_{\mc{V}_N} F_X(\xi)\,\mu_k(\d \xi)=\int_{\mc{V}_{N}} F_X(\p)
\,\tau_{\ast}\nu_{k,\overline{N}}(\d\p)
=\int_{\mc{V}_{\overline{N}}} F_X(\tau\psi) \,\nu_{k,\overline{N}}(\d\psi).
\end{align}
\end{proof}


\begin{proof}[Proof of Proposition \ref{proprenormmap}]
The proof relies on the bound \eqref{eq:finalderivbound} which we recall here. We consider a
smooth map
$M:(-\varepsilon,\varepsilon)\to \R^{s\times s}_{\mathrm{sym},+}$  and we denote $M=M(0)$ and $\dot{M}=\dot{M}(0)$. 
Then
\begin{align}\label{eq:finalderivbound2}
\left|\left.\frac{\d}{\d t} \int_{\mathbb{R}^s}  F(x) \,\mu_{M(t)}(\d x)\right|_{t=0}\right|\leq \lVert
F\rVert_{L^2(\mathbb{R}^s,\mu_{M(t)})}\lVert M^{-\frac{1}{2}}\dot{M}M^{-\frac{1}{2}}\rVert_{\mathrm{HS}}.
\end{align}
To see this we start from \eqref{firstderivative} and apply Cauchy-Schwarz 
\begin{align}\begin{split}
&\left|\left.\frac{\d}{\d t} \int_{\mathbb{R}^s}  F(x) \,\mu_{M(t)}(\d x)\right|_{t=0}\right|\\
&\qquad\leq \frac{1}{2}\left(\int_{\mathbb{R}^s}|F(x)|^2\,\mu_{M}(\d x)\right)^{\frac{1}{2}}
\left(\int_{\mathbb{R}^s}\left(\langle x, M^{-1}\dot{M}M^{-1}x\rangle-\tr M^{-\frac{1}{2}}\dot{M}M^{-\frac{1}{2}}\right)^2 \,\mu_{M}(\d
x)\right)^{\frac{1}{2}}.
\end{split}
\end{align}
Then the transformation $y=M^{-\frac12}x$ and an orthogonal transformation yields
\begin{align}
\int_{\mathbb{R}^s}\left(\langle x, M^{-1}\dot{M}M^{-1}x\rangle-\tr M^{-\frac{1}{2}}\dot{M}M^{-\frac{1}{2}}\right)^2 \,\mu_{M}(\d
x)=2\lVert M^{-\frac12}\dot{M}M^{-\frac12}\rVert_{\mathrm{HS}}^2.
\end{align}
Here the norm on the right hand side is the Hilbert Schmidt norm given
by $\lVert A\rVert_{\textrm{HS}}=\sqrt{\tr AA^T}$.
Lemma \ref{le:momentbound} below states a much more general version of this estimate.

We apply \eqref{eq:finalderivbound2} to the measures from Lemma \ref{lemmaloc}.
Set
 $M(t)=\ms{D}_{A+t\dot{A},k+1,k+1}$ and denote $\ms{\dot{D}}_{A,k+1,k+1}=\left.\frac{\d}{\d t}\ms{D}_{A+t\dot{A},k+1, k+1}\right|_{t=0}$ and
$\tilde{\ms{D}}=\ms{D}_{A,k+1,k+1}^{-\frac{1}{2}}
\ms{\dot{D}}_{A,k+1,k+1}\ms{D}_{A,k+1,k+1}^{-\frac{1}{2}}$. 
Combining this with Lemma \ref{lemmaloc} (here we need $F$ to be local) we get the estimate
\begin{align}
\left|D_A \int\limits_{\mc{X}_N} F(\p)\,\mu_{k}^{(A)}(\d\p)(\dot{A})\right|
=\left|D_A \int\limits_{\mc{V}_{k+1}}
F(\tau\psi)\,\nu_{k+1,k+1}^{(A)}(\d\psi)(\dot{A})\right|
\leq \lVert F\rVert_{L^2(\mc{V}_{k+1},\nu_{k+1,k+1}^{(A)})} \lVert\tilde{\ms{D}}\rVert_{\mathrm{HS}} 
\end{align}
where $\mu_{k+1}^{(A)}=\mu_{\ms{C}_{A,k+1}}$ and $\nu_{k+1,k+1}^{(A)}=\nu_{\ms{D}_{A,k+1,k+1}}$.
Note that reading Lemma \ref{lemmaloc} backwards implies 
$\lVert F\rVert_{L^2(\mc{V}_{k+1},\nu_{k+1,k+1}^{(A)})}
=\lVert F\rVert_{L^2(\mc{X}_{N},\mu_{k+1}^{(A)})}$.
The operators $\ms{D}_{A,k+1,k+1}$ are diagonal 
in Fourier space and satisfy by construction the equality
$\widehat{\mc{D}}_{A,k+1,k+1}(p)=\widehat{\mc{C}}_{A,k+1}(p)$ for
$p\in \widehat{T}_{k+1}\setminus\{0\}$ hence the Hilbert-Schmidt norm is given by
\begin{align}\label{eq:HSnormstart}
\lVert \tilde{\ms{D}}\rVert_{\mathrm{HS}}=\sum_{p\in \widehat{T}_{k+1}\setminus \{0\}}
\lVert\widehat{\mc{C}}_{A,k+1}(p)^{-\frac{1}{2}}\widehat{\mc{\dot{C}}}_{A,k+1}(p)\widehat{\mc{C}}_{A,k+1}(p)^{-\frac{1}{2}}\rVert^2_\mathrm{HS}.
\end{align}
Note that the Fourier mode for $p=0$ does not contribute because
it does not depend on $A$. Indeed $\mc{C}_{A,k}(x)=-M$ independent of $A$ for $|x|>L^{k+1}/2$
and $\widehat{\mc{C}}_{A,k}=0$ for all $A$.
We bound the expression in \eqref{eq:HSnormstart} using \eqref{keyquotientbound} and \eqref{Ahatestimate}
(denoting $\bd{A}_j=\bd{A}_j^{k+1}\subset \widehat{T}_{k+1}$)
\begin{align}
\begin{split}\label{eq:lastbound}
\sum_{p\in \widehat{T}_{k+1}\setminus \{0\}}
\lVert\widehat{\mc{C}}_{A,k+1}(p)^{-\frac{1}{2}}\widehat{\mc{\dot{C}}}_{A,k+1}(p)&\widehat{\mc{C}}_{A,k+1}(p)^{-\frac{1}{2}}\rVert_{\mathrm{HS}}
^2\\
&\leq 
m\sum_{p\in \widehat{T}_{k+1}\setminus \{0\}} \left(\lVert\widehat{\mc{C}}_{A,k+1}(p)^{-1}\rVert\lVert\widehat{\mc{\dot{C}}}_{A,k+1}(p)\rVert\right)^2\\
&\leq m\sum_{j=0}^{k+1} \sum_{p\in \bd{A}_j}
\myK^2L^{8(\tilde{n}+d)+4}L^{2(k+1-j)(n-\tilde{n})}  \\
&\leq
m\myK^2L^{8(\tilde{n}+d)+4}\sum_{j=0}^{k+1}|\bs{A}_j|L^{
(k+1-j)(2n-2\tilde { n })}\\
&\leq m\myK^2L^{8(\tilde{n}+d)+4}\sum_{j=0}^{k+1}\kappa(d)L^{(k+1-j)d}L^{
(k+1-j)(2n-2\tilde { n })}\\
&\leq C
\end{split}
\end{align}
where we used $2\tilde{n}-2n>d$ in the last step.
\end{proof}

We need a more general version of Proposition \ref{proprenormmap} to avoid the
loss of regularity in \cite{adams2016preprint}. Namely we must generalise the result
in Proposition \ref{proprenormmap} to higher order derivatives and we have to replace
the $L^2$ norm of $F$ on the right hand side of \eqref{eq:special} 
by a $L^p$ norm for any $p>1$. To understand the motivation for
this lemma we refer to the description of the renormalisation approach in the aforementioned work \cite{adams2016preprint}.
\begin{theorem}\label{prop:finalsmoothness}
  Let  $\ms{C}_{A,k+1}$  a finite range decomposition as in Theorem
\ref{finalFRD}
with $\tilde{n}-n>d/2$ and $X\subset T_N$ be a subset with diameter
$D=\mathrm{diam}_\infty(X)\geq
L^k$. Let $F:\mc{V}_N\rightarrow \mathbb{R}$ be a
functional that is measurable
with respect to the $\sigma$-algebra generated by $\{\p(x)|\, x\in X\}$, i.e.,
$F$ depends only on the values of the field $\p$ in $X$. Then for $\ell\geq 1$ and $p>1$ the following bound holds
\begin{align}
 \left|\frac{\d^\ell}{\d t^\ell}\int_{\mc{X}_N}\left. F(\p)\, \mu_{A+tA_1,k+
 1}(\d\p)\right|_{t=0}\right|\leq C_{p, L} (DL^{-k})^{\frac{d\ell}{2}} \lVert A_1\rVert^\ell
\lVert
F\rVert_{L^p(\mc{X}_N,\mu_{{A,k+1}})}.
\end{align}
\end{theorem}

\begin{proof}
 We use the notation from section \ref{secsetting} and the proof of Proposition
\ref{proprenormmap}. We first give explicit calculations for $\ell=1$ and $\ell=2$ and indicate the general case in the end.
 Recall in particular the definition of $\p_{M(t)}(x)$ in \eqref{eq:defofpM} and note that
 \begin{align}\begin{split}\label{eqfirsttwoderivs}
\left.\frac{\d}{\d t} \p_{M(t)}(x)\right|_{t=0}&=\frac{1}{2}\left(\langle
x,M^{-1}\dot{M}M^{-1}x\rangle-\tr\, M^{-1}\dot{M}\right)\p_{M}(x)\\
\left.\frac{\d^2}{\d t^2}
\p_{M(t)}(x)\right|_{t=0}&=\left[\frac{1}{4}\left(\langle
x,M^{-1}\dot{M}M^{-1}x\rangle-\tr M^{-1}\dot{M}\right)^2-
\langle x,M^{-1}\dot{M}M^{-1}\dot{M}M^{-1}x\rangle+\right. \\
&\;+\frac{1}{2}\left(\left.\langle x,M^{-1}\ddot{M}M^{-1}x\rangle-\tr
\left(M^{-1}\ddot{M}\right)\right)
+\frac{1}{2}\tr \left(M^{-1}\dot{M}M^{-1}\dot{M}\right)\right]\p_{M}(x).
\end{split}
\end{align}
We need the following general   lemma from 
\cite{Whittle:1960:BML}.
\begin{lemma}\label{le:momentbound}
 Let $X$ be a vector of $n$ independent standard normal variables and $A\in
\mathbb{R}^{n\times n}$ a matrix.
 Then for any real $s_0\geq 2$ there is a constant $C(s_0)$ such that for $1\leq
s\leq s_0$ the estimate 
 \begin{align}
 \mathbb{E}|\langle x,Ax\rangle -\tr A|^s\leq C(s_0)\lVert
A\rVert_{\textrm{HS}}^s
 \end{align}
 holds. 
 \end{lemma}
 \begin{proof}
 This is a special case of Theorem 2 in 
\cite{Whittle:1960:BML}. 
 The extension from $s\geq 2$ to $s\geq 1$ is a direct consequence of Hölder's inequality.
 \end{proof}
Using this lemma and the Hölder inequality with exponents $p$ and $p'$ we bound
\begin{align}\label{eq:firstderivativeFinal}
\begin{split}
 \left|\frac{\d}{\d t} \int_{\mathbb{R}^s}\left. F(x)\,
\mu_{M(t)}(\d x)\right|_{t=0}\right|&\leq
\frac12\lVert F\rVert_{L^p(\mu_M)}
\lVert \langle x,M^{-\frac12}\dot{M}M^{-\frac12}x\rangle-\tr
M^{-\frac12}\dot{M}M^{-\frac12}\rVert_{L^{p'}(\mu_\mathrm{id})} \\
&\leq C_p \lVert F\rVert_{L^p(\mu_M)}
\lVert M^{-\frac12}\dot{M}M^{-\frac12}\rVert_{\mathrm{HS}}
\end{split}
\end{align}

For the second derivative 
 we find similarly the bound 
 \begin{align}
 \begin{split}\label{eq:secondderivativeBound}
  &\left|\frac{\d^2}{\d t^2} \int_{\mathbb{R}^s}\left. F(x)\,
\mu_{M(t)}(\d x)\right|_{t=0}\right|\leq\\
&\quad\lVert F\rVert_{L^p(\mu_M)}
\left(\frac12|\tr M^{-\frac12}\dot{M}M^{-1}\dot{M}M^{-\frac12}|+ 
\right. \\
&\qquad \qquad \qquad \qquad+ 
\frac14\lVert(\langle x,M^{-\frac12}\dot{M}M^{-\frac12}x\rangle-\tr
M^{-\frac12}\dot{M}M^{-\frac12})^2\rVert_{L^{p'}(\mu_\mathrm{id})}+\\
&\qquad \qquad \qquad \qquad+\frac12\lVert \langle
x,M^{-\frac12}\ddot{M}M^{-\frac12}x\rangle-\tr
M^{-\frac12}\ddot{M}M^{-\frac12}\rVert_{L^{p'}(\mu_\mathrm{id})}+\\
&\qquad \qquad \qquad \qquad\left.+\lVert (\langle
x,M^{-\frac12}\dot{M}M^{-1}\dot{M}M^{-\frac12}x\rangle-\tr
M^{-\frac12}\dot{M}M^{-1}\dot{M}M^{-\frac12}\rVert_{L^{p'}(\mu_\mathrm{id})}
\right)
\end{split}
 \end{align}
where $p'$ denotes the Hölder conjugate of $p$.
The first trace term in \eqref{eq:secondderivativeBound}  can be bounded by 
Hölder inequality for
Schatten norms
(cf. Theorem 2.8 in \cite{MR2154153}) which yields
\begin{align}\label{eq:boundTraceTerm}
 |\tr M^{-\frac12}\dot{M}M^{-1}\dot{M}M^{-\frac12}|\leq
 \lVert M^{-\frac12}\dot{M}M^{-1}\dot{M}M^{-\frac12}\rVert_{\mathrm{Tr}}
\leq \lVert M^{-\frac12}\dot{M}M^{-\frac12}\rVert_{\mathrm{HS}}^2.
\end{align}
The trace norm is defined by $\lVert A\rVert_{\mathrm{Tr}}=\tr \sqrt{AA^\ast}$.
 With  the bound \eqref{eq:boundTraceTerm}, Lemma \ref{le:momentbound}, and the estimate $\lVert
AB\rVert_{\mathrm{HS}}\leq \lVert AB\rVert_{\mathrm{Tr}}\leq  \lVert
A\rVert_{\mathrm{HS}}\lVert
B\rVert_{\mathrm{HS}}$
 we conclude from \eqref{eq:secondderivativeBound}
 \begin{align}
  \begin{split}\label{eq:secondderivativeFinal}
  &\left|\frac{\d^2}{\d t^2} \int_{\mathbb{R}^s}\left. F(x)\,
\mu_{M(t)}(\d x)\right|_{t=0}\right|\leq \\
&\qquad\leq
C_p \lVert F\rVert_{L^p(\mu_M)}
\left(\lVert M^{-\frac12}\dot{M}M^{-\frac12}\rVert_{\mathrm{HS}}+\lVert
M^{-\frac12}\ddot{M}M^{-\frac12}\rVert_{\mathrm{HS}}+\lVert
M^{-\frac12}\dot{M}M^ {-\frac12}\rVert_{\mathrm{HS}}^2
\right).
\end{split}
 \end{align}
 For $\ell> 2$ the estimates are similar but more involved. We introduce some additional notation.
Let us state the general structure of $\left.\frac{\d^\ell}{\d t^\ell} \p_{M(t)}(x)\right|_{t=t_0}$.
We write $\left.\frac{\d^{\delta^i}}{\d t^{\delta^i}}M(t)\right|_{t=t_0}=M^{\delta_i}(t_0)$ for any positive integer $\delta^i$.
In the following we consider multiindices $\bs{\delta}=(\delta^1,\ldots,\delta^r)$ such that 
$r\geq 1$ and $\delta^i>0$ are integers. The length
of an index is denoted by $\myL(\bs{\delta})=r$, we call $|\bs{\delta}|_1=\sum_{i=1}^r \delta^i$ 
the degree of an index. With $I_\ell$ we denote the
set of such multiindices $\bs{\delta}$ such that $|\bs{\delta}|_1 \leq \ell$ and $r\geq 2$. 
Moreover $\overline{I}_\ell$ is
defined similarly with $r\geq 2$ replaced by $r\geq 1$. This distinction accounts for a cancellation
 for the $r=1$ terms.
We define for any $\bs{\delta}=(\delta_1,\ldots,\delta_r)$ and $r\geq 1$
\begin{align}
\myM_{\bs{\delta}}(t)=\prod_{i=1}^r M^{-\frac{1}{2}}(t)M^{\delta^i}(t)M^{-\frac{1}{2}}(t).
\end{align}

In the following three types of terms appear: 
\begin{align}\begin{split}
Q_{\bs{\delta}}(x,t)=\langle M(t)^{-\frac{1}{2}}x,\myM_{\bs{\delta}}(t) M(t)^{-\frac{1}{2}}x\rangle,\quad  
R_{\bs{\delta}}(t)=\tr
\myM_{\bs{\delta}}(t)\quad\text{with}\; \bs{\delta}\in I_\ell\;\;\text{and}\\
S_d(x,t)=\langle x,M(t)^{-1}M^d(t)M(t)^{-1}x\rangle-\tr
M(t)^{-1}M^d(t)\quad \text{with} \; 1\leq d\leq \ell .
\end{split}
\end{align}
The general expression can be written as
\begin{align}
\left.\frac{\d^\ell}{\d t^\ell} \p_{M(t)}(x)\right|_{t=t_0}=P_\ell(x,t_0)\p_{M(t_0)}(x)
\end{align}
where $P_\ell(x,t)$ is a linear combination of terms
\begin{align}\label{structureofsummands}
\prod_{i=1}^{k_1} Q_{\bs{\delta}_i}(x,t)
\prod_{i=1}^{k_2} R_{\bs{\epsilon}_i}(t)\prod_{i=1}^{k_3} S_{d_i}(x,t).
\end{align}
Here $\bs{\delta}_i,\bs{\epsilon}_i\in I_\ell$ and $1\leq d_i\leq \ell$
such that $\sum_{i=1}^{k_1}|\bs{\delta}_i|_1+
\sum_{i=1}^{k_2}|\bs{\epsilon}_i|_1+\sum_{i=1}^{k_3} d_i=\ell$, i.e. the total order of derivatives is $\ell$.
 From now on we drop the time argument and assume $t=0$.
 The explicit calculations for $\ell=1$ and $\ell=2$ showed that we need to bound
 the $p'$-norm of $P_\ell$. 
 Using Hölder's inequality with exponents $\frac{\ell}{|\bs{\delta}_i|_1}$, $\frac{\ell}{|\bs{\epsilon}_i|_1}$,
 and $\frac{\ell}{d_i}$ we estimate
 \begin{align}
 \begin{split}
    \left\lVert \prod_{i=1}^{k_1} Q_{\bs{\delta}_i}(x)
\prod_{i=1}^{k_2} R_{\bs{\epsilon}_i} \prod_{i=1}^{k_3} S_{d_i}(x)\right\rVert_{L^{p'}(\mu_{M})}
%
%
 &\leq \prod_{i=1}^{k_1} \lVert Q_{\bs{\delta}_i}(x)\rVert_{{\frac{p'\ell}{|\bs{\delta}_i|_1}}}
\prod_{i=1}^{k_2} \lVert R_{\bs{\epsilon}_i}\rVert_{\frac{p'\ell}{|\bs{\epsilon}_i|_1}}
 \prod_{i=1}^{k_3} \lVert S_{d_i}(x)\rVert_{\frac{p'\ell}{d_i}}
 \\
&\leq \max_{\bs{\delta}\in I_\ell} 
\lVert Q_{\bs{\delta}}(x)\rVert_{\frac{p'\ell}{|\bs{\delta}|_1}}^{\frac{\ell}{|\bs{\delta}|_1}}
\vee  \max_{\bs{\delta}\in I_\ell}
|R_{\bs{\delta }}|^{\frac{\ell}{|\bs{\delta}|_1}}) \vee 
\max_{1\leq d\leq \ell}
 \lVert S_{d}(x)\rVert_{{\frac{p'\ell}{d}}}^{\frac{\ell}{d_i}}
 \\
 &\leq 
 \max_{\bs{\delta}\in I_\ell} 
\lVert Q_{\bs{\delta}}(x)\rVert_{{p'\ell}}^{\frac{\ell}{2}}
\vee  \max_{\bs{\delta}\in I_\ell}
|R_{\bs{\delta }}|^{\frac{\ell}{2}} \vee 
\max_{1\leq d\leq \ell}
 \lVert S_{d}(x)\rVert_{{{p'\ell}}}^{{\ell}{}}\vee 1
 \end{split}
 \end{align}
 In the second step we used that $\prod x_i^{\alpha_i}\leq \max x_i$ if $\sum \alpha_i=1$.
 Then, in the third step, we exploited that
 the $L_p$ norms are monotone on probability spaces and $|\bs{\delta}|_1\geq \myL(\bs{\delta})\geq 2$ 
 for $\bs{\delta}\in I_\ell$.
 
We estimate the three terms with the help of Lemma \ref{le:momentbound}.
 For $1\leq d\leq \ell$
 \begin{align}\begin{split}\label{hsbound1}
 \int_{\R^s} |S_d(x)|^{{p'\ell}}\,\mu_M(\d x)&=
 \int_{\R^s} |\langle x,M^{-\frac{1}{2}}M^dM^{-\frac{1}{2}}x\rangle-\tr M^{-1}M^d|^{p'\ell}\,\mu_{\textrm{id}}(\d x)\\
& \leq C_{p'\ell}\lVert M^{-\frac{1}{2}}M^dM^{-\frac{1}{2}}\rVert_{\textrm{HS}}^{{p'\ell}{}}.
\end{split}
 \end{align}
Similarly we estimate the $Q$ terms
\begin{align}\begin{split}\label{hsbound2}
\int_{\R^s} |Q_{\bs{\delta}}(x)|^{{p'\ell }{}}\,\mu_M(\dx)
&=\int_{\R^s} |\langle y,\myM_{\bs{\delta}} y\rangle|^{{p'\ell }{}}\,
\mu_{\mathrm{id}}(\dy)\\
&\leq \int_{\R^s} 2^{{p'\ell}{}}\left(|\langle y,\myM_{\bs{\delta}} 
y\rangle-\tr \myM_{\bs{\delta}}|^{{p'\ell
}{}}+|\tr
\myM_{\bs{\delta}}|^{{ p'\ell }{}}\right)\,\mu_{\mathrm{id}}(\dy)\\
&\leq 2^{{p' \ell }{}}\left(C_{p'\ell}\lVert
\myM_{\bs{\delta}}\rVert_{\mathrm{HS}}^{{p' \ell }{}}
+|R_{\bs{\delta}}|^{{p' \ell}{}}\right).
\end{split}
\end{align}

For $\bs{\delta}\in I_\ell$ we find $\bs{\delta}_1, \bs{\delta}_2\in \bar{I}_\ell$
such that $\bs{\delta}=(\bs{\delta}_1,\bs{\delta}_2)$.
We bound the trace term for $\bs{\delta}$ using $\myM_{\bs{\delta}}=\myM_{\bs{\delta}_1}\myM_{\bs{\delta}_2}$
and the Hölder inequality for
Schatten norms
 by
\begin{align}\label{hsbound3}
|R_{\bs{\delta}}|=|\tr \myM_{\bs{\delta}}|
\leq \lVert \myM_{\bs{\delta}_1}\rVert_{\mathrm{HS}}\lVert \myM_{\bs{\delta}_2}\rVert_{\mathrm{HS}}
\leq \lVert \myM_{\bs{\delta}_1}\rVert_{\mathrm{HS}}^2\vee \lVert
\myM_{{\bs{\delta}_2}}\rVert_{\mathrm{HS}}^2.
\end{align} 
The estimates \eqref{hsbound1}, \eqref{hsbound2}, and \eqref{hsbound3}
imply that there is a constant $\bar{C}_{\ell,p}$ such that the following estimate holds

\begin{align}\label{hsboundfinal2}
\begin{split}
\left|\left.\frac{\d^\ell}{\d t^l} \int_{\R^s} F(x)\,\mu_{M(t)}(\d x)
\right|_{t=0}\right|&\leq
\lVert F\rVert_{L^p(\mu_M)}\left(\int_{\R^s} |P_\ell(x)|^{p'}\,\mu_M(\d x)\right)^{\frac{1}{p'}}
\\
&\leq \bar{C}_{\ell,p} \lVert F\rVert_{L^p(\mu_M)}\left( \max_{\bs{\delta}\in
\bar{I}_\ell} \lVert
\myM_{\bs{\delta}}\rVert_{\mathrm{HS}}^{\ell}\vee 1\right).
\end{split}
\end{align}
 After these  technical estimates the rest of the proof is 
 similar to the proof of Proposition \ref{proprenormmap}.
We consider
$M(t)=\ms{D}_{A+t\dot{A},k+1,\overline{N}}$ where $\overline{N}$ is the
smallest integer such that $L^{\overline{N}}\geq 2D$ (or $\overline{N}=N$ if
$D\geq L^N/2$).
Moreover we again denote 
$\ms{{D}}_{A,k+1,k+1}^r=\left.\frac{\d^r}{\d t^r}\ms{D}_{A+t\dot{A},k+1,
\overline{N}}\right|_{t=0}$ and for $\bs{\delta}\in I_\ell$
\begin{align}
 \ms{M}_{\bs{\delta}}
 =\prod_{i=1}^{\myL(\bs{\delta})} \ms{D}_{A,k+1,\overline{N}}^{-\frac{1}{2}}
\ms{{D}}^{\delta_i}_{A,k+1,\overline{N}}\ms{D}_{A,k+1,\overline{N}}^{-\frac{1}{2}}.
\end{align}
Using Lemma \ref{lemmaloc} and \eqref{hsboundfinal2} we bound
\begin{align}
\begin{split}\label{eq:boundwithF}
\left|D^\ell_A \int_{\mc{X}_N} F(\p)\mu_{k+1}^{(A)}(\d\p)(\dot{A},\ldots,\dot{A})\right|
&=\left|D^\ell_A \int_{\mc{V}_{\overline{N}}}
F(\tau\psi)\nu_{k+1,\overline{N}}^{(A)}(\d\psi)(\dot{A},\ldots,\dot{A})\right|\\
&\leq C\lVert F\rVert_{L^p(\mc{V}_{\overline{N}},\nu_{k+1,\overline{N}}^{(A)})} 
\left(\max_{\bs{\delta}\in I_\ell}\lVert\ms{M}_\delta\rVert_{\mathrm{HS}}^\ell
\vee 1 \right)
\\
&= C\lVert F\rVert_{L^p(\mc{X}_{N},\mu_{k+1}^{(A)})} 
\left(
\max_{\bs{\delta}\in I_\ell}\lVert\ms{M}_\delta\rVert_{\mathrm{HS}}^\ell\vee 1\right)
\end{split}
\end{align}
where $\mu_{k+1}^{(A)}=\mu_{\ms{C}_{A,k+1}}$ and $\nu_{k+1,\overline{N}}^{(A)}=\nu_{\ms{D}_{A,k+1,\overline{N}}}$.
For the last identity we used  Lemma \ref{lemmaloc} backwards.

It remains to estimate the Hilbert-Schmidt norm of the operators $\ms{M}_{\bs{\delta}}$.
The operators $\ms{M}_{\bs{\delta}}$ are diagonal 
in Fourier space and  by construction we have the equality
$\widehat{\mc{D}}_{A,k+1,\overline{N}}(p)=\widehat{\mc{C}}_{A,k+1}(p)$ for
$p\in \widehat{T}_{\overline{N}}$ and all $A$ hence we can bound
the operator norm of the  Fourier modes of $\ms{M}_{\bs{\delta}}$ for $\lVert \dot{A}\rVert \leq 1$
using \eqref{keyquotientbound} for $p\in \bs{A}_j$ and $j<k$ as follows
\begin{align}
\begin{split}
 \lVert \widehat{\ms{M}}_{\bs{\delta}}(p)\rVert&\leq \prod_{i=1}^{\myL(\bs{\delta})} \lVert\widehat{\mc{D}}_{k,\bar{N}}^{-1}(p)\rVert\cdot
\lVert \widehat{\mc{D}}_{k,\bar{N}}^{\delta_i}(p)\rVert\\
&\leq \prod_{i=1}^{\myL(\bs{\delta})}\myK(\delta_i)L^{4(\tilde{n}+d)+2}L^{(k-j)(n-\tilde{n})}
\\
&\leq \bar{\myK}(\bs{\delta})^2L^{4 \myL(\bs{\delta})(\tilde{n}+d)+2\myL(\bs{\delta})} 
L^{(k-j)(n-\tilde{n})}.
\end{split}
\end{align}
Here we wrote $\bar{\myK}(\bs{\delta})=\prod_{j=1}^{\myL(\bs{\delta})} \myK(\delta_j)$ for the product of the constants.
Similar, for $p\in A_j$ with $j\geq k$, 
\begin{align}\label{eqfourierestimates13}
\lVert \widehat{\ms{M}}_\delta(p)\rVert&\leq  \bar{\myK}(\bs{\delta})L^{2\myL(\bs{\delta})(\tilde{n}+d)+\myL(\bs{\delta})}.
\end{align}
Then the Hilbert-Schmidt norm is bounded by 
(with $\bd{A}_j=\bd{A}_j^{\overline{N}}\subset \widehat{T}_{\overline{N}}$)
\begin{align}
\begin{split}
\lVert \ms{M}_{\bs{\delta}}\rVert_{\mathrm{HS}}^2
&=\sum_{p\in
\widehat{T}_{\overline{N}}\setminus \{0\}} \lVert \widehat{\ms{M}}_{\bs{\delta}}(p)\rVert_{\mathrm{HS}}^2
\leq m\sum_{p\in
\widehat{T}_{\overline{N}}\setminus \{0\}} \lVert \widehat{\ms{M}}_{\bs{\delta}}(p)\rVert^2
\\
&\leq m\sum_{j=0}^{k-1} \sum_{p\in \bd{A}_j}
\bar{\myK}^s2L^{8\ell(d+\tilde{n})+4}L^{2(k-j)(n-\tilde{n})} 
+m\sum_{j=k}^{\overline{N}} \sum_{p\in \bd{A}_j}
\bar{\myK}^2L^{4\ell(d+\tilde{n})+2}\\
&\leq
CL^{8\ell(d+\tilde{n})+4}\sum_{j=0}^{k-1}L^{(\overline{N}-j)d}L^{
(k-j)(2n-2\tilde{n})}
+CL^{4\ell(d+\tilde{n})+2}\sum_{j=k}^{\overline{N}}L^{(\overline{N}-j)d}
\\
&\leq 
C L^{8\ell(d+\tilde{n})+4} L^{(\overline{N}-k)d}.
\end{split}
\end{align}
where we used $2\tilde{n}-2n>d$ in the last step.
The bound \eqref{eq:lastbound} and $2D>L^{\overline{N}-1}$ imply that 
\begin{align}
\lVert \ms{M}_{\bs{\delta}}\rVert_{\mathrm{HS}}^\ell
\leq CL^{\frac{\ell}{2}(\overline{N}-k)d}< C (DL^{-k})^{\frac{\ell d}{2}}.
\end{align}
Plugging this into \eqref{eq:boundwithF} ends the proof.

\end{proof}

\begin{appendix}\label{app:A}
 \section{Proof of Theorem \ref{AKMFRD}}
 In this appendix we discuss those details of the proof of Theorem \ref{AKMFRD} that are not already
 contained in the much more general discussion in \cite{MR3129804}.
 The key ingredient of the proof is the following lemma
 \begin{lemma}[Lemma 2.3. in \cite{MR3129804}]
  Let $B>0$ be a constant. There is  a smooth family of functions
$W_t\in C^\infty(\mathbb{R})$ for $t>0$
  such that for  $\lambda\in (0,B)$, $t>0$ and integers $\ell,n\geq0$ the
following holds
\begin{align}\label{eq:keyWproperty}
  \lambda^{-1}=\int_{0}^\infty tW_t(\lambda)\,\d t, \\
  W_t(\lambda)\geq 0 \\ \label{eq:keyWestimate}
(1+t^2\lambda)^n\lambda^\ell\left|\frac{\partial^\ell}{\partial\lambda^\ell}
W_t(\lambda)\right|\leq C_{\ell,n}.
\end{align}
Moreover $W_t$ is a polynomial of degree at most $t$. 
In addition we can choose $W_t$ such that it satisfies
\begin{align}\label{eq:Wlowerbound}
 W_t(\lambda)\geq \epsilon
\end{align}
for some $\epsilon>0$ and  $\lambda\leq B\min(1,t^{-2})$.
 \end{lemma}
\begin{proof}
 This is Lemma 2.3. in \cite{MR3129804} (with rescaled $\lambda$) except for the
lower
bound \eqref{eq:Wlowerbound}. The lower bound is however easily obtained from the construction in
\cite{MR3129804}.
One possible choice for the function $W_t$ is
given by
\begin{align}
 W_t(\lambda)=\sum_{n\in \mathbb{Z}}
\p\left(\arccos\left(1-\frac{\lambda}{2B}\right)t-2\pi nt\right)
\end{align}
for $\lambda\in (0,B)$ where $\p:\mathbb{R}\rightarrow \mathbb{R}^+_0$ is any symmetric positive
function such that $\widehat{\p}$ is supported in $[-1,1]$ and smooth.
More precisely, we let $\p=|\kappa|^2$ where $\widehat{\kappa}$ is even
and supported in $[-\frac12,\frac12]$ and we can moreover choose it to be
non-negative, which implies for $|x|<2$ that 
\begin{align}
 \kappa(x)=\frac{1}{2\pi}\int_\mathbb{R} \widehat{\kappa}(k)e^{ikx}\,\d
x=\frac{1}{2\pi}\int_{-\frac12}^{\frac12} \widehat{\kappa}(k)\cos(kx)\,\d x>\sqrt{\epsilon}
\end{align}
for some $\epsilon>0$.
The bound $\arccos(1-x)\leq \frac{\pi}{2}\sqrt{2x}$ for $x\in[0,2]$ implies
that for $\lambda<B\min(t^{-2},1)$ the estimate $|t\arccos(1-\frac{\lambda}{2B})|<2$
holds.
Hence for those $\lambda$ we bound
\begin{align}
 W_t(\lambda)\geq
\p\left(\arccos\left(1-\frac{\lambda}{2B}
\right)t\right)=\left|\kappa\left(\arccos\left(1-\frac { \lambda } { 2B }
\right)t\right)\right|^2\geq
\epsilon.
\end{align}
\end{proof}
\begin{proof}[Proof of Theorem \ref{AKMFRD}]
 We set $B=\pi^2d\Omega\geq \max(\mathrm{spec}(\ms{A}))$
 where $\Omega$ is the constant in \eqref{Ahatestimate}.
Based on the previous lemma we obtain a finite range decomposition by defining
for
$2\leq k\leq N$
\begin{align}\begin{split}
  \ms{C}_k&=\int_{\frac{L^{k-1}}{2R}}^\frac{L^k}{2R} tW_t(\ms{A})\,\d t\\
 \ms{C}_1&=\int_0^\frac{L}{2R} tW_t(\ms{A})\,\d t\\
  \ms{C}_{N+1}&=\int_{\frac{L^{N}}{2R}}^\infty tW_t(\ms{A})\,\d t
 \end{split}
\end{align}
This decomposition indeed satisfies $\sum_{k=1}^{N+1}\ms{C}_k=\ms{C}$ because
$\mathrm{spec}(\ms{A})\subset [0,B]$ by \eqref{Ahatestimate} and property \eqref{eq:keyWproperty}.
Since $W_t$ is a polynomial of degree at most $t$ and
$\mathrm{supp}({\ms{A}\p})\subset \mathrm{supp}(\p)+[-R,\ldots,R]^d$ the finite
range property in the theorem holds.

Next, we want to show that the matrix $-M_k=\mc{C}_{A,k}(x)$ with $|x|\geq L^k/2$ is positive definite and independent of $A$.
The kernel $\mc{C}_{A,k}$ is uniquely characterised by the conditions that $\ms{C}_{A,k}\p=\mc{C}_{A,k}\ast \p$ for $\p \in \mc{X}_N$
and $\mc{C}_{A,k}\in \mc{M}_N$ (space of matrix valued kernels with average zero).
By construction of $\ms{C}_{A,k}=\ms{C}_k(\ms{A})$ we know that for $1\leq k\leq N$ there are coefficients 
$c_{k,l}$ such that $\ms{C}_k(\ms{A})=\sum c_{k,l}\ms{A}^l$.
Observe that the action of the finite difference operator $\ms{A}$ can be written as
\begin{align}
 \ms{A}\p(x)=\sum_{y\in [-R,R]^d} a_y \p(x+y)
\end{align}
where $a_y\in \R^{m\times m}$ are coefficients such that $\sum_y a_y=0$.
In particular the kernel  $\mc{A}\in \mc{M}_N$ of $\ms{A}$ satisfies $\mc{A}(x)=0$ for
$d_\infty(x,0)>R$. The same holds for 
powers $\ms{A}^l$ and  $d_\infty(x,0)>lR$ because the kernel of $\ms{A}^l$
is given by the $l$-fold convolution $\mc{A}\ast\ldots \ast \mc{A}$.
Therefore only the constant term of the polynomial  contributes to the kernel of 
$\mc{C}_{A,k}(x)$ for $|x|\geq L^k/2$ in particular it is independent of $\ms{A}$. 
Note that the kernel of  $\ms{A}^0=\mathrm{id}_{\mc{X}_N}$  is given by $(\delta_0-L^{-Nd})\mathbb{1}_{m\times m}$ because this function has average zero and
constant shifts do not change the operator so it generates the same operator as $\delta_0 \mathbb{1}_{m\times m}$.
In order to show $M_k\geq 0$ it is therefore sufficient to 
show $c_{k,0}\geq 0$ because $M_k=-c_{k,0}L^{-Nd}\mathrm{Id}$.
The inequality $W_t(0)\geq 0$ implies that the constant term
of the polynomial $W_t$ is positive.   Hence
\begin{align}
c_{k,0}=\left(\int_{L^{k-1}/(2R)}^{L^k/(2R)} t\cdot W(0)\,\d t\right)>0.
\end{align}

It remains to establish the bounds. Here it is useful to rely on the estimate \eqref{Ahatestimate}
for the spectrum of the Fourier coefficients instead of the bounds on the quadratic form $Q$.
Since $\ms{A}$ is diagonal in Fourier space and this property carries over to
polynomials in $\ms{A}$ the identity 
\begin{align}
 \widehat{\mc{C}}_k(p)=\int_{\frac{L^{k-1}}{2R}}^\frac{L^k}{2R}
t\cdot W_t(\widehat{\mc{A}}(p))\,\d t
\end{align}
holds for $2\leq k\leq N$ and similar identities hold for $k=1$ and $k=N+1$.
Using this equation we can derive strong bounds for the Fourier modes of
$\ms{C}_k$. 
Let us denote the eigenvalues of the symmetric and positive matrix
$\widehat{\mc{A}}(p)$ by $\omega|p|^2\leq\lambda_1\leq \ldots \leq\lambda_m$
where we plugged in the lower bound \eqref{Ahatestimate}.
The key observation is that by estimate \eqref{eq:keyWestimate}
\begin{align}\label{eq:OperatorNormW}
 \lVert W_t(\widehat{\mc{A}}(p))\rVert = \max_{1\leq i\leq m}
|W_t(\lambda_i)|\leq  \frac{C_n}{(\lambda_1 t^2)^n}\leq
\frac{C_n}{(\omega |p|^2 t^2)^n}.
\end{align}
The estimate \eqref{eq:OperatorNormW} implies with $n\geq 2$ and $n=0$
respectively
\begin{align}\label{eq:Ck1}
 \lVert\widehat{\mc{C}}_k(p)\rVert&\leq
\int_{\frac{L^{k-1}}{2R}}^{\infty} \frac{C_n}{(\omega |p|^2t^2)^n}t\, \d
t\leq
\frac{C_n(2R)^{2n-2}}{(2n-1)\omega^n}|p|^{-2}(|p|L^{k-1})^{-(2n-2)}\;\forall \, 2\leq k\leq N+1\\
\label{eq:C k2}
\lVert\widehat{\mc{C}}_k(p)\rVert&\leq
\int_{0}^{\frac{L^k}{2R}} C_0t\, \d
t\leq \frac{C_0L^{2k}}{8R^2}\; \forall 1\leq k\leq N.
\end{align}
Note that the first bound does not hold for $k=1$ and the last bound does not hold
for $k=N+1$. Moreover there is the trivial bound 
\begin{align}\label{eq:Cktriv}
 \lVert\widehat{\mc{C}}_k(p)\rVert\leq \lVert\widehat{\mc{C}}(p)\rVert
 =\lVert\widehat{\mc{A}}(p)^{-1}\rVert\leq
\frac{|p|^{-2}}{\omega}.
\end{align}
The most useful combination of these bounds is
\begin{align}\label{eq:RegimeBound}
 \lVert\widehat{\mc{C}}_k(p)\rVert\leq 
\left\{\begin{alignedat}{2}
  & C_{n'}|p|^{-2}(|p|L^{(k-1)})^{-n'} \quad &&\text{for}\,
|p|\geq L^{-(k-1)},\; n'\geq 2, \\
  & C|p|^{-2}\quad &&\text{for}\, L^{-(k-1)}>|p|\geq L^{-k},\\
  & CL^{2k}\quad &&\text{for}\, L^{-k}>|p|.
  \end{alignedat}\right.
\end{align}
This bound also holds for $k=1$ and $k=N+1$. Indeed we have
$(|p|L^{k-1})^{-n'}\geq (\pi \sqrt{d})^{-n'}$ for $k=1$
so \eqref{eq:Cktriv} implies the first estimate and
$\{p\in\widehat{T}_N:\,|p|<L^{-(N+1)}\}=\emptyset$. 
In particular \eqref{eq:RegimeBound} contains \eqref{fourierbounds} for $\ell=0$.
For $|\alpha|+d>2$ we get using  \eqref{eq:RegimeBound} with $n'>d-2+|\alpha|$
and
 \eqref{l1linftyFourier}   
\begin{align} 
\begin{split}\label{eq:NablaCk}
 \lVert \nabla^\alpha \mc{C}_k(x)\rVert 
 &\leq \frac{C}{L^{Nd}}
\left(\sum_{\substack{p\in\widehat{T}_N\\|p|\geq L^{-(k-1)}}}\hspace{-0.2cm}
|p|^{|\alpha|-2}(|p|L^{(k-1)})^{-n'}+\hspace{-0.4cm} 
\sum_{\substack{p\in\widehat{T}_N\\
 L^{-(k-1)}>|p|\geq L^{-k}}}\hspace{-0.3cm} |p|^{|\alpha|-2}
 +\sum_{\substack{p\in\widehat{T}_N\\
L^{-k}>|p|}}\hspace{-0.2cm} |p|^{|\alpha|}L^{2k}\right) \\
&\leq C\hspace{-0.4cm}\int \limits_{\scriptstyle{L^{-(k-1)}\leq
r}}\hspace{-0.4cm} r^{d-3+|\alpha|}
(rL^{(k-1)})^{-n'}\,\d r+C\hspace{-0.6cm}
\int\limits_{L^{-k}<r<L^{-(k-1)}} \hspace{-0.6cm}r^{d-3+|\alpha|}
\,\d r+C
\int\limits_{r<L^{-k}} \hspace{-0.2cm}L^{2k}r^{|\alpha|+d-1}\,\d r\\
&\leq CL^{-(k-1)(d-2+|\alpha|)}+CL^{-k(d-2+|\alpha|)}\leq
CL^{-(k-1)(d-2+|\alpha|)}
\end{split}
\end{align}
 where we approximated the Riemann sums by their
integrals possibly increasing the constant. This approximation can be justified
using a dyadic decomposition for the sums.
Note that the constant $C$ does not depend on $N$ or $L$.
The condition $d+|\alpha|>2$ was used to bound the second integral. For
$d+|\alpha|=2$ it behaves as $\int \frac{\d r}{r}\approx
\ln(L)$ hence we get an additional logarithm in this case.

Next we consider derivatives with respect to the parameter matrix $A$.
We need the following simple lemma for which we did not find an exact reference
in the
literature. Similar arguments can be found in \cite{MR2814377}.
\begin{lemma}\label{le:deriv}
 Let $A, B\in \mathbb{R}^{m\times m}_{\mathrm{sym}}$ be symmetric matrices and
let $\lambda_1\leq \ldots\leq \lambda_m$ be
the  eigenvalues of
$A$ counted with multiplicity. Let $f$ be a holomorphic function. Then there is
a combinatorial constant
$C_{m,\ell}$ such that
\begin{align}
 \left\Vert \left.\frac{\d^\ell}{\d s^\ell} f(A+sB)\right|_{s=0}\right\Vert\leq
C_{m,\ell}
\sup_{\lambda\in [\lambda_{1},\lambda_{m}]} |f^{(\ell)}(\lambda)| \lVert B\rVert^{\ell}.
\end{align}

\end{lemma}
\begin{proof}
The proof is based on a representation of the matrix derivative using the
Cauchy formula that appears e.g. in \cite{MR1335452}.
Note that the eigenvalues of $A+sB$ are continuous functions of $s\in
\mathbb{R}$ for $A$
and $B$ symmetric \cite{MR1335452}. 
Let $C$ be a curve around all the eigenvalues of $A+sB$ for $s\in(-\epsilon,
\epsilon)$ with winding number 1. 
By the Cauchy formula
\begin{align}
 f(A+sB)=\frac{1}{2\pi i}\int_C f(z) (z\mathrm{Id}-(A+sB))^{-1}
\end{align}
Differentiating $\ell$ times with respect to $s$ or using the Neumann series for the matrix inverse  gives
\begin{align}\label{eq:matrixderivative}
\left. \frac{\d^\ell}{\d s^\ell}f(A+sB)\right|_{s=0}=\frac{\ell!}{2\pi i}\int_C
f(z)
(z\mathrm{Id}-A)^{-1}B(z\mathrm{Id}-A)^{-1}B\ldots B(z\mathrm{Id}-A)^{-1}.
\end{align}
Now we write $A=\sum_{i=1}^m \lambda_iP_i$ as a sum of orthogonal projections
such that $\sum_{i=1}^m P_i=\mathrm{Id}$.
Then we find for $z\notin \mathrm{spec}(A)$ that
$(z\mathrm{Id}-A)^{-1}=\sum_{i=1}^n (z-\lambda_i)^{-1}P_i$.
Plugging this in \eqref{eq:matrixderivative} we bound
\begin{align}\begin{split}
\left\Vert\left. \frac{\d^\ell}{\d
s^\ell}f(A+sB)\right|_{s=0}\right\Vert&=\left\Vert\frac{\ell!}{2\pi
i}\sum_{i_1, \ldots, i_{\ell+1}=1}^m\int_C
\frac{f(z)}{(z-\lambda_{i_1})\ldots(z-\lambda_{i_n})}P_{i_1}BP_{i_2}\ldots
BP_{i_{\ell+1}}\right\Vert\\
&\leq\sum_{i_1, \ldots, i_{\ell+1}=1}^m \left|\frac{\ell!}{2\pi
i}\int_C
\frac{f(z)}{(z-\lambda_{i_1})\ldots(z-\lambda_{i_{\ell+1}})}\right| \lVert
B\rVert^\ell
\end{split}
\end{align}
The term in absolute values is the sum of divided differences \cite{MR2221566} and by
the mean
value theorem for finite difference (Proposition 43 in \cite{MR2221566})
there is a $\xi\in (\min_{1\leq j\leq n+1}\lambda_{i_j},
\max_{1\leq j\leq n+1}\lambda_{i_j})$ such that
\begin{align}
 \frac{\ell!}{2\pi
i}\int_C
\frac{f(z)}{(z-\lambda_{i_1})\ldots(z-\lambda_{i_\ell})}
=f^{(\ell)}(\xi).
\end{align}
This implies the claim
\begin{align}
 \left\Vert \left.\frac{\d^\ell}{\d s^\ell} f(A+sB)\right|_{s=0}\right\Vert\leq
m^{(\ell+1)}
\sup_{\lambda\in [\lambda_1,\lambda_m]} |f^{(\ell)}(\lambda)|\lVert
B\rVert^\ell.
\end{align}
\end{proof}
We apply this lemma to the Fourier modes of the operators $\ms{A}_{A+sA_1}$
where $A_1\in \mc{L}(\mc{G})$ is a linear and symmetric but not necessarily
positive
operator.
Note that $W_t$ can be extended to a holomorphic function in a neighbourhood
of $(0,B)$ because this holds for the $\arccos$ function and $\p$ is holomorphic
since $\widehat{\p}$ has compact support.
Then we find using linearity of the Fourier transform, Lemma \ref{le:deriv}, the bounds for the
spectrum of the Fourier modes \eqref{Ahatestimate}, and the estimate
\eqref{eq:keyWestimate}
\begin{align}\begin{split}
 \left\Vert \frac{\d^\ell}{\d
s^\ell}W_t(\widehat{\mc{A}}_{A+sA_1}(p))\right\Vert
 &=
 \left\Vert \frac{\d^\ell}{\d
s^\ell}W_t(\widehat{\mc{A}}_{A}(p)+s\widehat{\mc{A}}_{A_1}(p))\right\Vert\\
&\leq 
C_{m,\ell}\lVert
\widehat{\mc{A}}_{A_1}(p)\rVert^\ell \sup_{\lambda\in
\mathrm{Conv}(\mathrm{spec}
\;\widehat{\mc{A}}_{A}(p))} |W^{(\ell)}(\lambda)| \\
&\leq 
C_{m,\ell}\Omega^\ell |p|^{2\ell}\lVert
A_1\rVert^\ell \sup_{\lambda\in \mathrm{Conv}(\mathrm{spec}\;
\widehat{\mc{A}}_{A}(p))} |W^{(\ell)}(\lambda)|\\
&\leq 
C_{m,\ell} \Omega^\ell\lVert
A_1\rVert^\ell |p|^{2\ell} \sup_{\lambda\in \mathrm{Conv}(\mathrm{spec}\;
\widehat{\mc{A}}_{A}(p))}\frac{C_{n,\ell}}{\lambda^\ell (1+\lambda t^2)^n}
\\
&\leq 
C_{m,n,\ell}\left(\frac{\Omega}{\omega}\right)^{\ell}\lVert A_1\rVert^\ell  \min(1, (\omega
|p|^2t^2)^{-n}).
\end{split}
\end{align}
This is up to a constant exactly the same bound we used before to obtain
\eqref{eq:Ck1} (for $k\geq 2$) and \eqref{eq:C k2} (for $k\leq N$).
Let us check that also the bound $\lVert \widehat{\mc{C}}_k(p)\rVert<C|p|^{-2}$
generalises to $\ell\geq 1$.
We bound for $L^{-k}<|p|<L^{-(k-1)}$ and $\lVert \dot{A}\rVert \leq 1$
\begin{align}
 \lVert D^\ell_A \widehat{\mc{C}}_k(p)(\dot{A},\ldots,\dot{A})\rVert\leq 
 \int_{\frac{L^{k-1}}{2R}}^{|p|^{-1}} C_{m,0,\ell}\cdot t\,\d t+
 \int_{|p|^{-1}}^{\frac{L^{k}}{2R}} \frac{C_{m,2,\ell}}{ |p|^{4}t^{4}}t\,\d t\leq
C|p|^{-2}
\end{align}
Hence we find a bound similar to \eqref{eq:RegimeBound} 
\begin{align}
 \label{eq:RegimeBound2}
 \lVert D^\ell_A\widehat{\mc{C}}_k(p)(\dot{A},\ldots, \dot{A})\rVert\leq 
\left\{\begin{alignedat}{2}
  & C_{n'}|p|^{-2}(|p|L^{(k-1)})^{-n'} \quad &&\text{for}\,
|p|\geq L^{-(k-1)},\; n'\geq 2 \\
  & C|p|^{-2}\quad &&\text{for}\, L^{-(k-1)}>|p|\geq L^{-k}\\
  & CL^{2k}\quad &&\text{for}\, L^{-k}>|p|
  \end{alignedat}\right.
\end{align}
This bound completes the proof of the upper bound in Fourier space \eqref{fourierbounds}.
As in \eqref{eq:NablaCk} this bound also implies \eqref{discretebounds} for
$\ell\geq 1$.

Finally we consider the lower bound. 
For $|p|\leq t^{-1}$ we find $\lVert \widehat{\mc{A}}(p)\rVert \leq \Omega
|p|^2\leq Bt^{-2}$. Hence \eqref{eq:Wlowerbound} implies
 $W_t(\widehat{\mc{A}}(p))\geq \epsilon
\mathbb{1}_{m\times m}$. 
Using this bound we find for $|p|\leq L^{-k}$ and $k\geq 2$
\begin{align}
 \widehat{\mc{C}}_k(p)=\int_{\frac{L^{k-1}}{2R}}^{\frac{L^k}{2R}}
tW_t(\widehat{\mc{A}}(p))\,\d t\geq \frac{\epsilon L^{2k}}{16R^2}.
\end{align}
Using the positivity of $W_t$ we get for $k=1$
\begin{align}
 \widehat{\mc{C}}_k(p)&=\int_{0}^{\frac{L}{2R}}
tW_t(\widehat{\mc{A}}(p))\,\d t\\
&\geq \epsilon \int_{0}^{\min(\frac{L}{2R},|p|^{-1})}
t\,\d t \geq \frac{\epsilon}{2}\min\left(|p|^{-2},\frac{L^2}{4R^2}\right).
\end{align}
This completes the proof of Theorem \ref{AKMFRD}.
\end{proof}

\end{appendix}

	\section*{Acknowledgements}
	This work was supported by the Bonn International Graduate School of
Mathematics and the CRC 1060
	{\it The mathematics of emergent effects}.
	 I would like to thank my
advisor S. M\"{u}ller for his advice and several helpful discussions and  David Brydges for several helpful comments.
	\bibliographystyle{amsplain}
\addcontentsline{toc}{section}{References}
\bibliography{/home/buchholz/Dokumente/PHDLiteratur/Literatur/New_PHD.bib}
\textsc{Institute for Applied Mathematics, University of Bonn,
Endenicher Allee 60, 53115 Bonn}\newline
E-mail: buchholz@iam.uni-bonn.de
\end{document}